\documentclass[11pt]{article}

\usepackage[shortlabels]{enumitem}
\usepackage{amsmath} 
\usepackage{amssymb} 
\usepackage{mathrsfs} 
\usepackage{mathtools} 
\usepackage{amsthm}
\usepackage{bm}
\usepackage{thmtools}
\usepackage{thm-restate}
\usepackage{appendix}
\usepackage[in]{fullpage}
\usepackage{caption}
\usepackage{subcaption}
\usepackage{graphicx}
\usepackage[table,xcdraw]{xcolor}
\usepackage{wrapfig}
\usepackage{floatrow}
\usepackage{multirow}
\usepackage{makecell}
\usepackage{nicefrac}
\usepackage{dsfont}

\usepackage[ruled,vlined,linesnumbered]{algorithm2e}
\definecolor{forestgreen}{rgb}{0.13, 0.55, 0.13}

\SetCommentSty{mycommfont}
\SetKwInOut{Input}{input}
\SetKwInOut{Output}{output}
\DontPrintSemicolon

\definecolor{ForestGreen}{rgb}{0.0333,0.4451,0.0333}
\definecolor{DarkRed}{rgb}{0.65,0,0}
\definecolor{Red}{rgb}{1,0,0}
\usepackage[linktocpage=true,
pagebackref=true,colorlinks,
linkcolor=DarkRed,citecolor=ForestGreen,
bookmarks,bookmarksopen,bookmarksnumbered]{hyperref}

\usepackage{cleveref}
\usepackage{thm-restate}
\usepackage{tcolorbox}
\usepackage[normalem]{ulem} 

\newtheorem{theorem}{Theorem}[section]
\newtheorem{lemma}[theorem]{Lemma}

\newtheorem{remark}[theorem]{Remark}
\newtheorem{corollary}[theorem]{Corollary}
\newtheorem{proposition}[theorem]{Proposition}

\newtheorem{definition}[theorem]{Definition}

\newcommand*\mb[1]{\mathbb{#1}}

\newcommand*\mc[1]{\mathcal{#1}}
\DeclarePairedDelimiter\abs{\lvert}{\rvert}
\DeclarePairedDelimiter\norm{\lVert}{\rVert}%
\newcommand{\iprod}[1]{\left\langle{#1}\right\rangle}

\makeatletter
\let\oldabs\abs
\def\abs{\@ifstar{\oldabs}{\oldabs*}}
\let\oldnorm\norm
\def\norm{\@ifstar{\oldnorm}{\oldnorm*}}
\makeatother

\newcommand{\ts}{\textstyle}
\newcommand{\eps}{\varepsilon}
\let\epsilon\varepsilon

\newcommand{\R}{\mathbb{R}}
\newcommand{\ALG}{\mathsf{ALG}}
\newcommand{\OPT}{\mathsf{OPT}}

\DeclareMathOperator*{\argmax}{arg\,max} 
\DeclareMathOperator*{\argmin}{arg\,min}

\DeclareMathOperator{\spn}{span}
\newcommand{\dual}{\text{Dual}}

\DeclareMathOperator{\rank}{rank}

\newcommand{\alg}[1]{\widehat{#1}}

\newcommand{\spec}[1]{\widehat{#1}}

\newcommand{\polym}{\mc{P}_f}
\newcommand{\parm}{\mc{Q}}
\newcommand{\pa}{Q}

\newcommand{\Rp}{\mb{R}_{\geq 0}}
\renewcommand\vec[1]{\bm{#1}}
\newcommand{\up}[1]{^{(\vec{#1})}}
\newcommand{\upn}[1]{^{({#1})}}

\newcommand{\ealt}{{e_{\text{alt}}}}
\newcommand{\one}{\mathds{1}}

\newif\ifcomments
\commentstrue

\ifcomments
\usepackage[colorinlistoftodos,prependcaption,textsize=tiny]{todonotes}

\definecolor{dnotecol}{rgb}{0.20, 0.50, 0.80}
\newcommand{\dnote}[1]{{\textcolor{dnotecol}{\textbf{DH:} #1}}}
\newcommand{\dnotein}[1]{\todo[linecolor=red,backgroundcolor=dnotecol!25,bordercolor=red,inline]{\textbf{DH:~}#1}}

\definecolor{electricviolet}{rgb}{0.56, 0.0, 1.0}
\definecolor{mulberry}{rgb}{0.77, 0.29, 0.55}
\definecolor{olive}{rgb}{0.5, 0.5, 0.0}

\definecolor{kgreen}{rgb}{0.0, 0.7, 0.24}

\definecolor{imAboutIt}{rgb}{1, 0.5, .31}

\else 

\newcommand{\dnote}[1]{}
\newcommand{\dnotein}[1]{}
\newcommand{\dqtodo}[1]{}
\newcommand{\dqtodoin}[1]{}
\newcommand{\rnote}[1]{}
\newcommand{\rtodo}[1]{}

\fi 

\newcommand*\samethanks[1][\value{footnote}]{\footnotemark[#1]}

\title{The Online Submodular Assignment Problem}
\author{Daniel Hathcock\thanks{Carnegie Mellon University, D.~Hathcock supported by the NSF GRFP grant DGE-2140739.} \and  Billy Jin\thanks{Cornell University} \and Kalen Patton\thanks{Georgia Institute of Technology, Supported in part by NSF award CCF-2327010.} \and Sherry Sarkar\samethanks[1] \and Michael Zlatin\samethanks[1]}
\date{}

\begin{document}

\maketitle

\begin{abstract}

Online resource allocation is a rich and varied field. One of the most well-known problems in this area is  online bipartite matching, introduced in 1990 by Karp, Vazirani, and Vazirani \cite{KVV90}. Since then, many variants have been studied, including AdWords, the generalized assignment problem (GAP), and online submodular welfare maximization. 

In this paper, we introduce a  generalization of GAP which we call the submodular assignment problem (SAP). This generalization captures many online assignment problems, including all classical online bipartite matching problems as well as broader online combinatorial optimization problems such as online arboricity, flow scheduling, and laminar restricted allocations. We present a fractional algorithm for online SAP that is $(1-\frac{1}{e})$-competitive. 

Additionally, we study several integral special cases of the problem. In particular, we provide a $(1-\frac{1}{e}-\epsilon)$-competitive integral algorithm under a small-bids assumption, and a $(1-\frac{1}{e})$-competitive integral algorithm for online submodular welfare maximization where the utility functions are given by rank functions of matroids.

The key new ingredient for our results is the construction and structural analysis of a “water level” vector for polymatroids, which allows us to generalize the classic water-filling paradigm used in online matching problems. This construction reveals connections to submodular utility allocation markets and principal partition sequences of matroids.
\end{abstract}

\pagenumbering{gobble}
\newpage

\pagenumbering{arabic}
\setcounter{page}{1}

\section{Introduction}
\label{sec:introduction}

Online assignment problems are 
fundamental in the study of online algorithms. Perhaps the most well-known online assignment problem is online bipartite matching, introduced by Karp, Vazirani, and Vazirani \cite{KVV90}. In online bipartite matching, we are given one side of a bipartite graph (the \emph{offline} vertices) in advance, while the vertices on the other side arrive online. When an online vertex arrives, all of its incident edges are revealed, and the algorithm selects at most one of the edges. The goal is to maximize the number of edges chosen, subject to the edges being a matching in the graph. For this problem, Karp, Vazirani, and Vazirani proposed the Ranking algorithm which achieves a tight $1 - \nicefrac{1}{e}$ competitive ratio.

Since then, online bipartite matching has received considerable attention, and more general variations of the problem have been studied. Some of the most prominent examples include:
\begin{itemize}
    \item \textbf{Vertex and Edge Weighted Variants.} In vertex weighted online bipartite matching, each \emph{offline} vertex has a weight, and the goal is to maximize the sum of the weights of the matched offline vertices. In the more general edge-weighted setting, individual edges have weight and the goal is to maximize the sum of the weights of the selected edges. 
    \item \textbf{AdWords.} This was introduced by Mehta, Saberi, Vazirani, and Vazirani~\cite{MSVV07}, motivated by the AdWords market in digital advertising. Each offline vertex $i$ has a budget $B_i$ and each edge $e$ has a bid $b_e$. Selecting an edge consumes an amount of budget from the offline vertex equal to the bid of the edge. The goal is to maximize the total sum of the bids of the selected edges, subject to the budget constraints. Note that vertex-weighted bipartite matching is a special case of AdWords with $b_{ij} = B_i$ for all edges $ij$.
    \item \textbf{Generalized Assignment Problem (GAP).} Here, every offline vertex has a budget $B_i$, and every edge $e$ has both a value $v_e$ and a cost $b_e$. The goal is to maximize the total value of the selected edges, such that the total cost of the edges incident to any offline vertex does not exceed its budget. This is one of the broadest online matching problems that has been studied in the literature, and generalizes all of the settings above\footnote{We consider GAP in the setting of \cite{FKMMP09}, which includes the small-bids and free disposal assumptions. In other literature that considers GAP as an offline problem, these assumptions are not usually made. It is only with these assumptions that GAP can be considered a generalization of AdWords.}. In particular, AdWords is the special case of GAP with $b_e = v_e$ for all edges $e$. Edge-weighted bipartite matching is the case with all $B_i$ and $b_e$ equal to 1.
\end{itemize}

All of the above problems admit $(1-\nicefrac{1}{e})$-competitive algorithms under various assumptions. For vertex-weighted bipartite matching, Aggarwal, Goel, Karande and Mehta \cite{AggarwalGKM11} give a generalization of the Ranking algorithm which is $(1-\nicefrac{1}{e})$-competitive. For AdWords, \cite{MSVV07} show the same competitive ratio can be achieved for the fractional version of the problem, and more generally for the integral version under a small-bids assumption.\footnote{The small-bids assumption states that the ratio $\frac{b_e}{B_i}$ should be small, for any offline $i$ and any edge $e$ incident to $i$.} Edge-weighted bipartite matching and GAP are commonly studied under the ``free disposal" assumption, which is necessary to outmaneuver a trivial $\nicefrac{1}{n}$ hardness in these settings. Under free disposal, $(1-\nicefrac{1}{e})$-competitive algorithms can be obtained for fractional edge-weighted bipartite matching and GAP, and for GAP under a small-bids assumption \cite{FKMMP09}. 


Nevertheless, many natural online assignment problems exist which are not captured by the above settings. We illustrate these in the examples below. To our knowledge, no optimally competitive algorithms for these problems are implied by prior work. 

\begin{itemize}
    \item \textbf{{Laminar} Restricted Matchings.} Consider the AdWords problem. Suppose that in addition to the budget constraints for each offline node, we have a laminar family $\mathcal{S}$ of subsets of offline nodes, and there is a budget constraint for each $S \in \mathcal{S}$. For instance, this can model a setting where a company has several departments each with their own individual ad budget, and the company as a whole also has an additional budget constraint for the total amount that can be spent across all its departments. 
    
    \item \textbf{Matroid Coloring.} Suppose we have a matroid $\mathcal{M}$ whose elements arrive one by one online. We have $\Delta$ colors and may irrevocably assign a color to each element as it arrives, subject to the constraint that each color must be independent in $\mathcal{M}$. The objective is to color as many elements as possible. Two natural applications of this problem are: 
    \begin{itemize}
        \item \textit{Online Arboricity.} Suppose the edges of an undirected graph $G = (V, E)$ arrive online. When each edge arrives, we irrevocably decide whether or not to select it. The goal is to maintain the largest possible sub-graph with arboricity\footnote{The arboricity of a graph is the minimum number of forests required to cover its edges.} at most $\Delta$. One way to solve this problem is by modelling it as a matroid coloring problem, where $\mathcal{M}$ is the graphic matroid associated with $G$. The arboricity of a graph is a well studied property which has been used to maintain dynamic edge orientations \cite{CCHHQRS23} and proper colorings of a sub-graph \cite{CR22}.
        
        \item \textit{Flow Scheduling.} Suppose we have a network $N$, with integer capacities on the edges, that is known up front with a single sink $t$. 
        The times where the network is available for use is partitioned into $\Delta$ many time slots. Source vertices with unit demand appear one by one. When a source $s_j$ appears, we must schedule it in one of the time slots (or not schedule it at all). The goal is to maximize the  number of assignments, such that for every time slot, it is feasible to simultaneously send the flow for all sources scheduled in that slot. This is matroid coloring where $\mathcal{M}$ is a gammoid. 
    \end{itemize}

    \item \textbf{Coflows.} Say we have a computing resource which may process some tasks in parallel. For example, perhaps a single server rack is made up of different servers, each of which is equipped to handle only certain types of tasks. What a server rack can handle is modelled via a bipartite graph, with potential tasks on one side and servers of the server rack on the other. Tasks which may be processed together on a single rack form a transversal matroid; these are called \textit{coflows}, inspired by applications to MapReduce \cite{ChowdhuryS12}. Coflows governed by general matroid constraints have been studied \cite{JahanjouKR17, ImMPP19} in an offline setting. In an online formulation of this problem, we have $\Delta$ server racks and $n$ tasks arriving online. The tasks are splittable, but have different costs and values for being completed at different servers (i.e., some servers are closer or cheaper than others). We must irrevocably split tasks among computing resources, though we my drop tasks later on; the goal is to handle as many tasks as possible.


\end{itemize}

In this paper, we define the \textit{Online Submodular Assignment Problem}, which captures all of the problems described earlier as special cases. Via our results on this more general problem, we provide $1-\nicefrac{1}{e}$ competitive algorithms for all the problems above. 

\subsection{Problem Statement}

The \textit{Online Submodular Assignment Problem} (\textit{Online SAP} for brevity), is as follows. We have an (offline) monotone submodular\footnote{A function $f: 2^E \rightarrow \Rp$ is submodular if for all $A, B \subseteq E$, we have $f(A \cup B) \leq f(A) + f(B) - f(A \cap B)$. It is monotone if $f(A) \leq f(B)$ whenever $A \subseteq B$.} function $f$ over ground set $E$ with $f(\varnothing) = 0$ and $f(\{e\}) > 0$ for all $e \in E$.\footnote{This assumption is without loss of generality, since any $e$ with $f(\{e\}) = 0$ can be removed.} Every element $e \in E$ has a value $v_e$ and a cost $b_e$. The ground set, initially unknown, is partitioned into parts $\pa_1, \dots, \pa_m$ that arrive online one-by-one. Upon arrival, each $\pa_j$ reveals its contained elements along with their values and costs. We have offline access to an evaluation oracle for $f$ that may be called on any subset of elements revealed so far. 

When a part $Q_j$ arrives, we may select at most one element from $Q_j$. At any point, we also can choose to freely dispose of elements previously selected (known as the \textit{free disposal} assumption).\footnote{For our results this assumption is not used in settings where $v_e = b_e$ for all $e$, generalizing results for AdWords.} The goal is to choose a set $S^* \subseteq E$ so as to maximize $\sum_{e \in S^*} v_e$ while maintaining that $S^*$ satisfies the online assignment constraints
\[ 
    \abs{S^* \cap Q_j} \leq 1 \quad \text{for all $j \in \{1, \hdots, m\}$}
\]
and the offline submodular constraints
\[ 
    \sum_{e \in S} b_e \leq f(S) \quad \text{for all $S \subseteq S^*$} 
\] 
We note that, since online SAP is a generalization of edge weighted online bipartite matching, the free disposal assumption is necessary to avoid a trivial $\nicefrac{1}{n}$-hardness. 

In the fractional variant of this problem, we instead choose a fractional allocation $(x_e)_{e \in \pa_j}$ on the elements in $\pa_j$ when it arrives. In accord with the free disposal assumption, we may decrease $x_e$ at any point. The objective is to maximize the final value of $\sum_{e \in E} v_e x_e$. As before, we must allocate no more than 1 total unit to elements in each $\pa_j$. In other words, we have $\vec{x}(\pa_j) := \sum_{e \in \pa_j} x_e \leq 1$. Moreover, the total cost vector $\vec{bx} := (b_e x_e)_{e \in E}$ must obey submodular constraints defined by $f$, i.e., so that $\vec{bx}(S) = \sum_{e \in S} b_e x_e \leq f(S)$ for every $S \subseteq E$. Put another way, we must maintain a point $x \in \polym \cap \parm$, where $\polym$ and $\parm$ are defined respectively as: 



\[
    \polym := \left\{\vec{x} \in \Rp^E: \vec{bx}(S) \leq f(S) \text{ for every } S \subseteq E \right\}
\]

and 
\[
    \parm := \left\{\vec{x} \in \Rp^E : \vec{x}(\pa_j) \leq 1 \text{ for every } j = 1, \ldots, n\right\}.
\]

Note that Online SAP captures all three assignment problems posed in the introduction. We show how the Laminar Restricted Matching Problem can be modeled as Online SAP in \Cref{sec:laminar-adwords-osap}. In Online Matroid Coloring, to color a matroid $\mathcal{M}$ online with $\Delta$ colors, we consider the product matroid $\mathcal{M}^\Delta := \mathcal{M} \times \hdots \times \mathcal{M}$ and define the submodular constraint to be the rank function of the lifted matroid $f := \text{rank}_{\mathcal{M}^\Delta}$. The assignment constraint dictates each element may map to at most 1 color, and the submodular constraint $f := \text{rank}_{\mathcal{M}^\Delta}$ enforces that each color remains an independent set in $\mathcal{M}$. The third problem is a version of weighted matroid coloring, where elements have different valuations for different colors. 


\subsection{Our Contributions}\label{sec:intro-results-tech}

We introduce the online submodular assignment problem, which encompasses many online assignment problems. Some of these are well-known, including vertex- and edge-weighted bipartite matching, AdWords, and GAP. Others, such as laminar restricted matching, matroid coloring, and coflow assignment, have not been solved previously. Not only do we get optimal competitive ratios for online SAP in several settings, but in doing so we develop a novel framework for handling submodular constraints in online assignment. We consider the development of this machinery to be the primary contribution of our work, as we believe it may be broadly useful for future applications to problems with similar structure.

Our first main theorem concerns the fractional version of online SAP. 

\begin{theorem}\label{thm:online-SAP}
    There exists a deterministic $(1 - \nicefrac{1}{e})$-competitive algorithm for the fractional Online Submodular Assignment problem. 
\end{theorem}
We note that the $(1 - \nicefrac{1}{e})$ competitive ratio is tight, as there is a matching upper bound even for the special case of fractional online bipartite matching \cite{feige2020tighter}.

Next, we show that our fractional algorithm can be adapted to an integral algorithm under a ``small bids'' assumption. This type of assumption is often made in the AdWords setting~\cite{MSVV07,FKMMP09}, where the costs (or ``bids") are assumed to be small compared to the budgets of the advertisers. Specifically, the cost $b_e$ of an offline vertex is assumed to at most a $\varepsilon$-fraction of the total budget of its offline vertex for some $\epsilon > 0$. 
We note that it is still an open problem to determine the optimal competitive ratio for integral AdWords without the assumption of small bids; only recently have researchers developed an algorithm that achieves a competitive ratio better than $\nicefrac{1}{2}$ \cite{huang2020adwords}.

Our next result concerns the integral version of online SAP under a small bids assumption about the marginal functions $f_T(S) := f(S \cup T) - f(T)$. This generalizes the small bids assumption for AdWords. 
\begin{restatable}[Small Bids]{assumption}{smallbids}
    \label{ass:small_bids}
 Assume there exists some $\epsilon > 0$ such that for all $e \in E$ and $T \subseteq E$ with $f_T(\{e\}) > 0$, we have $b_e \leq \epsilon f_T(\{e\})$. 
\end{restatable} 


\begin{restatable}{theorem}{thmsmallbids}
    \label{thm:small-bids}
    Assume $b_e = v_e$ for all $e \in E$. Then, under the small bids assumption (\Cref{ass:small_bids}), there is a deterministic integral algorithm for online SAP which is $(1-O(\epsilon)) \cdot \left(1 - \frac{1}{e}\right)$-competitive. 
\end{restatable}


In addition to our result in the small bids setting, we also obtain $1 - \nicefrac{1}{e}$ competitiveness in a special case of the integral setting without the need for a small bids assumption. Suppose that the submodular function $f$ is the rank function of a cross-product matroid $ \mathcal{M} := \mathcal{M}_1 \times \hdots \times \mathcal{M}_n$, and each part $\pa_j$ contains at most one element from each $\mathcal{M}_i$. Then, assuming $b_e = v_e = 1$ for all $e$, this case of Online SAP is equivalent to the Online Submodular Welfare Maximization problem where the agents have matroid rank valuations. 

The Online Submodular Welfare Maximization Problem (OSWM), is the problem of assigning $m$ indivisible items, which arrive online, to $n$ agents with utility functions $f_i: 2^{[m]} \to \Rp$. Each utility function $f_i$ is assumed to be a monotone, submodular function on $[m]$. The goal is to find an assignment $\sigma: [m] \to [n]$ that maximizes the total welfare of the agents, i.e., the sum of the utilities $\sum_{i = 1}^n f_i(\sigma^{-1}(i))$. In the offline setting, a $1 - \nicefrac{1}{e}$ approximation algorithm for OSWM can be achieved using an algorithm for Monotone Submodular Maximization subject to a matroid constraint~\cite{CalinescuCPV11}. Surprisingly however, Kapralov, Post, and Vondr\'ak ~\cite{kapralovPV2013online} show that, in the online setting, achieving an competitive ratio greater than $\nicefrac{1}{2}$ in polynomial time ($\nicefrac{1}{2}$ is achieved trivially with the Greedy algorithm) is impossible unless NP = RP. By providing an integral algorithm for Online SAP in this special case, we are able to show that this barrier can be circumvented if the class of monotone submodular utility functions is restricted to rank functions of a matroid over ground set $[m]$.

\begin{theorem}\label{thm:matroid-oswm}
There exists a randomized polynomial time $(1-\nicefrac{1}{e})$-competitive algorithm for Online Submodular Welfare Maximization when the agents' utility functions are matroid rank functions. 
\end{theorem}

This setting captures (integral) online matroid coloring. Indeed, given a matroid $\mathcal{M}$ we seek to color online, each item represents an edge and each agent represents a color class. Thus, our result implies a $(1-\frac{1}{e})$-competitive algorithm for integral online matroid coloring.



\subsection{Technical Overview}\label{sec:tech-overview}

We primarily examine online submodular assignment in the fractional setting, as the same techniques yield results for the integral setting with small bids. For this problem, we will employ a continuous pricing-based allocation, in which we place prices $p_e$ on elements $e \in E$ and allocate continuously to the element maximizing utility $v_e - p_e$. This approach is similar to algorithms by \cite{MSVV07}, \cite{FKMMP09}, and \cite{KP2000}, which have yielded $(1-\nicefrac{1}{e})$-competitive algorithms for a wide range of online matching problems, including edge weighted online bipartite matching, AdWords, and GAP. Each of these can be interpreted as general forms of the classic water filling algorithm for online bipartite matching.

In standard online matching, the water-filling algorithm keeps track of a water level for each offline vertex, which is the fractional amount it has been filled so far. Upon the arrival of each online vertex, the algorithm continuously sends water to its neighbors with the current lowest water level, until either all neighbors have been filled or the arriving vertex has been depleted.

\paragraph{Submodular Water Levels} The key challenge in obtaining a water-filling type algorithm for online SAP is adapting the concept of ``water levels" from prior work on matchings, to a general submodular constrained linear program. Intuitively, given a fractional allocation $\{x_e\}_{e \in E}$, we would like to assign a water level $w_e$ to element $e$ which represents how ``filled" it is under the allocation $\vec x$. This water level is then used to determine the price of each element, with a higher water level leading to a higher price. 


In online bipartite matching, the water level of an edge $(i,j)$ (where $i$ is offline and $j$ is online), is the fraction of vertex $i$'s capacity occupied by $\vec x$, namely $\sum_{j' \sim i} x_{ij'}$. In the submodular assignment setting, edges become elements $e \in E$, but there is no longer a notion of unit-capacity vertices. Suppose momentarily that $b_e = 1$ for all $e$, for ease of presentation. Then each $e$ is involved in many submodular constraints of the form $\vec{x}(S) \leq f(S)$ for each $S \ni e$, and it is not immediately clear how these constraints should be aggregated into a single water level $w_e$.

Nevertheless, we show that there exist such ``submodular water levels'' that can be computed efficiently, and that they satisfy several key properties that make a water-filling type algorithm possible. We present three different ways of describing the submodular water levels. 
\begin{enumerate}[1.]
    \item \textbf{A combinatorial description.} Water levels can be viewed as a measure of density. We describe an iterative process in \Cref{alg:comb-dec} which finds a densest set under $\vec x$ (i.e. the set with the least multiplicative slack in its constraint), contracts this set, and repeats until all elements are contracted. The density at which an element gets contracted is the water level of that element. This algorithm also demonstrates that the level sets of $\vec{w}$ correspond to a weighted version of the \textit{principal partition sequence} studied in \cite{Narayanan91, Fujishige09} for poly-matroids.
        
    \item \textbf{A max-min relation.} The water level of an individual element $e$ can also be described directly via a simple max-min formula $\max_{S \ni e} \min_{T} \frac{\vec x(S\setminus T)}{f_T(S)}$, where $f_T$ denotes the \textit{contraction} of $f$ by $T$. This definition exhibits a surprising minimax theorem showing that the water level is also $\min_{T \not \ni e} \max_S \frac{\vec x(S\setminus T)}{f_T(S)}$, and moreover, a ``saddle point" pair of optimizing sets $S^*, T^*$ for both of these expressions can be derived from the combinatorial description of water levels. Importantly, this definition immediately demonstrates that the water level of each element is a continuous piecewise linear function of $\vec x$. 
    
    \item \textbf{A convex program.} The concept of water levels also arise from market equilibria. Jain and Vazirani \cite{JV10} introduced a notion of \textit{submodular utility allocation (SUA) markets} as a generalization of linear Fisher markets  \cite[Chapter 5]{AGT}. In a submodular utility allocation market, we have some items $A$, each with weight $m_a$. We would like to fractionally select a set of items, with the goal of maximizing the Nash social welfare of the items; however, the set of items picked must be feasible in a poly-matroid defined by a submodular function $f$. 
    It turns out the water level vector is precisely the optimal solution of an SUA market.
\end{enumerate}

Once we have our notion of water levels defined, we can define an algorithm for online SAP which is inspired by prior work in matching settings. At a high level, when a part $Q_j$ arrives, we assign a utility to each of its elements depending on its value and current water level. We make a small increase to the allocation of the highest utility element, and continue this process until we can no longer increase any element while preserving feasibility.

The second key technical challenge in performing a water-filling type analysis is understanding how the water levels change throughout the algorithm. We use a primal-dual framework, as in \cite{DevanurJK12}, to analyze the approximation ratio of our algorithm. In this framework, the algorithm maintains a primal solution at every step, and for the sake of analysis, we continuously maintain a dual solution in parallel. This dual solution will depend on the current water levels of the primal allocation. Critically, we must prove that throughout the algorithm, the dual value is at most the primal value and the dual solution is approximately feasible. To perform such analysis, we need a strong structural understanding of the water level vector $\vec w$ and how it changes with $\vec x$. We show that the vector $\vec w$ is remarkably well-behaved, and each of our three viewpoints on water levels sheds light on the properties that we will use in our analysis:
\begin{itemize}
    \item \textbf{Monotonicity and Continuity.} For each $e \in E$, its water level $w_e$ is monotone and continuous in the allocation vector $\vec{x}$ (\Cref{cor:wl-mono}).
    \item \textbf{Feasibility Indication.} The water level vector $\vec{w}$ indicates the feasibility of allocation vector $\vec{x}$ (\Cref{cor:wl-feasibility}). In particular, $\vec{x}$ is feasible if and only if the water levels are at most 1.
    \item \textbf{Locality.} Changing the allocation $x_e$ of some element by a small amount can only affect the the water level of elements whose water level is equal to $e$'s water level.
    
    \item \textbf{Duality.} The water levels satisfy a duality property with the Lov\'asz extension of $f$ (\Cref{lem:wl-duality}). That is $L_f(\vec{w}) = \sum_{e \in E} x_e$. This property can be derived from the KKT conditions of the convex program formulation of water levels.
\end{itemize} 
These properties allow us to keep track of how the dual objective of SAP changes with incremental changes of the primal and dual assignments. Specifically, careful setting of the dual ensures that approximate feasibility is maintained. Moreover, we prove a derivative formula for the Lov\'asz extension of $f$ in terms of the current water levels (\Cref{lem:L-G-derivative}) which allows us to upper bound the rate of change of the dual objective by that of the primal.



\paragraph{OSWM with Matroid Rank Valuations}
Our algorithm for Online Submodular Welfare Maximization for matroidal rank valuations is a natural generalization of the classical RANKING algorithm for online bipartite matching. We interpret the arriving nodes as items which we assign to the offline nodes (the agents). At the very beginning of the algorithm, we randomly permute the agents. When an item arrives, we assign it to the ``available" agent  of highest priority. 

In the bipartite matching case, and even for online $b$-matching (studied in~\cite{AlbersS21}), whether an agent is available simply depends on how many items have been allocated to this agent so far. In our setting, we say that an item is available to an agent if the agent's marginal utility for the item is non-zero, given the items already allocated to this agent. Hence, the availability of an agent depends on the item being considered and changes throughout the course of the algorithm.

Despite this, we show that our Matroidal RANKING algorithm always yields a tight $(1-\frac{1}{e})$-competitive allocation in expectation. Similar to the analysis in~\cite{DevanurJK12}, we construct an approximately feasible dual solution which lower bounds the welfare of our allocation. The main technical hurdle is to define a suitable threshold $r^*_{ij}$ for each agent-item pair $(i,j)$ so that the dual assignment (which is constructed online) is approximately feasible, and whose objective value matches the welfare of our online assignment.

Finally, we show that our algorithm can be extended in a natural way to the case where each agent has a non-negative weight, and we seek to maximize the weighted sum of their utilities.



\subsection{Related Work}

\paragraph{Online Matching} There is an extensive line of work on online matching, starting with the work of Karp, Vazirani and Vazirani \cite{KVV90}, who gave a $(1-\nicefrac{1}{e})$-competitive algorithm for online bipartite matching in the adversarial order setting. The same competitive ratio was extended to the vertex-weighted setting in \cite{AggarwalGKM11}, and further to the vertex-weighted $b$-matching setting in \cite{AlbersS21}. Devanur, Jain, and Kleinberg~\cite{DevanurJK12} showed how the results in \cite{KVV90} and \cite{AggarwalGKM11} could be derived using the online primal-dual framework, which unified and simplified the existing analyses. While the RANKING algorithm requires $O(n \log n)$ bits of randomness, Buchbinder, Naor, and Wajc~\cite{BNW23} provide a randomized rounding scheme requiring only $(1 \pm o(1)) \log \log n$ bits of randomness. 

For edge-weighted online bipartite matching, \cite{FahrbachHTZ20} were the first to break the $\frac{1}{2}$-competitive barrier. This has been subsequently refined in \cite{ShinA21, GaoHHNYZ21, BlancC21} to a 0.5368-competitive ratio. Online edge-weighted bipartite matching has also been studied in the vertex-arrival Bayesian setting, also known as the \emph{Ride Hail} problem; \cite{PPSW21} give a better-than-$\nicefrac{1}{2}$ approximation to the optimal online policy. See also~\cite{feldman2009},~\cite{Hwang2022}, and~\cite{Shuyi2024} for other Bayesian formulations which compare to the optimal offline solution. Further generalizations of online bipartite matchings have been studied, where the \textit{online} constraint is replaced with general allocation constraints \cite{DHKMY16}.

\paragraph{Online Matroid Matching} Closely related to our setting is the work of Wang and Wong~\cite{WangWong2016} who studied online bipartite matching under matroidal constraints on the matched offline nodes. This captures the special case of online SAP where $b_e = v_e = 1$ for all $e \in E$. They give a $(1-\frac{1}{e})$-competitive algorithm for the fractional version of this problem, and for the integral case assuming random-order arrivals. Zhang and Conitzer~\cite{HanruiConitzer2020} study the same model but with matroidal constraints on both the offline and online nodes, giving the same optimal $(1- \frac{1}{e})$-competitive guarantees. 

While these papers also generalize the  classic water-filling approach to accommodate submodular constraints, our work crucially differs in that (1) we allow non-uniform values $v_e$ and costs $b_e$, and (2) we greatly simplify the primal-dual proofs through the development of our general framework of water levels. 


\paragraph{Online Matroid Intersection} Offline matroid intersection captures a broad class of combinatorial problems and has been well studied in both polyhedral theory and algorithms (for a survey on this topic, see \cite[Chapter 41]{Schrijver03}). In \cite{GS17}, Guruganesh and Singla study an online matroid intersection problem, with edges arriving in random order, beating the $\nicefrac{1}{2}$-competitive ratio achieved by Greedy. This was very recently improved by Huang and Sellier \cite{HS23} to $\nicefrac{2}{3} + \varepsilon$. In \cite{BGHKS23}, the part-arrival online matroid intersection model is considered. They instead study the problem of \textit{maintaining} a max cardinality independent set, while minimizing recourse.

\paragraph{Offline Submodular Welfare Maximization} The offline variant of Submodular Welfare Maximization problem can be cast as a problem of maximizing a monotone submodular function subject to a (partition) matroid constraint. While a simple greedy algorithm achieves an approximation ratio of $\frac{1}{2}$, \cite{CalinescuCPV11} gave an improved $(1-\nicefrac{1}{e})$-approximation using the Continuous Greedy algorithm followed by Pipage Rounding. There can be no $(1-\nicefrac{1}{e}+\varepsilon)$-approximation unless P=NP since this problem captures max-$k$-cover as a special case~\cite{feige1998threshold}. 

When the monotone submodular function is the rank function of a matroid, the welfare maximizing offline allocation can be computed optimally. These allocations have been well-studied in the context of designing socially optimal allocations which satisfy certain desirable fairness constraints~\cite{viswanathan2023general, dror2023fair}. For example, in~\cite{benabbou2021finding} it is shown that an optimal allocation which is envy-free up to one item (EFX)  exists and can be computed efficiently. In~\cite{babaioff2021fair} they study the case of valuations which have binary marginals, but which are not necessarily submodular~\cite{babaioff2021fair}. They also give truthfulness guarantees for private valuations, which has been further studied in~\cite{barman2022truthful}. 

\paragraph{Principal Partition} 
The \textit{principal partition} of a matroid is related to our definition of water levels. It is precisely the nested sets found in \Cref{alg:comb-dec}, when $\vec x$ is the all ones vector. There have been several works studying principal partitions, including generalizations to arbitrary vectors $\vec x \in \mathbb{R}^E$ and extensions from rank functions to submodular functions \cite{Narayanan91, Fujishige09}. The objects studied in \cite{Narayanan91} and \cite{Fujishige09} are closely related to our water levels; in our work, we (1) shift perspective from the family of nested sets in \cite{Fujishige09} to the properties of a single vector $\vec w$, and (2) study properties of how $\vec w$ changes dynamically with the weights given by $\vec x$, such as monotonicity (\Cref{cor:wl-mono}), duality (\Cref{lem:wl-duality}), and locality (\Cref{lem:wl-change}). These properties are clearly visible with our new perspective. We make the novel connection between principal partitions and water-filling in online bipartite matching. 

The principal partition has been used for constant competitive algorithms in the matroid secretary problem under random assignment \cite{Soto13}. Huang and Sellier \cite{HS23} also use the principal partition to get an improved approximation ratio of $\nicefrac{2}{3} - \varepsilon$ for online matroid intersection with random order arrivals in a stream. Chandrasekaran and Wang \cite{CW23} use the principal partition sequence for improved approximations in the submodular $k$-partition problem. 

\subsection{Organization of the Paper}
First, in \Cref{sec:prelims}, we introduce some preliminaries. Next, in \Cref{sec:water-filling}, we develop a theory of water levels and provide three equivalent formulations. In \Cref{sec:fractional}, we look at a special and illuminating case of online SAP: online matroid intersection with part arrivals. We give a water-filling algorithm for fractional online poly-matroid intersection and prove that it achieves a competitive ratio of $1-\nicefrac{1}{e}$. \Cref{sec:gap} gives the general case analysis for fractional online SAP, as well as online SAP with a small bids assumption. In \Cref{sec:oswm}, we give a $(1-\nicefrac{1}{e})$-competitive algorithm for (integral) Online Submodular Welfare Maximization with matroidal utilities.

\label{sec:Related-Work}

\section{Preliminaries}
\label{sec:prelims}







We analyze the algorithm for fractional online SAP via the primal-dual framework. The primal linear program (LP) describing Online Submodular Assignment, as well as its dual are:
\begin{align*}
\max  \quad &\sum_{e \in E}  v_ex_e  &\qquad  \min \quad &\sum_{S \subseteq E} f(S)\cdot \alpha_S + \sum_{j=1}^n \beta_j \\
\mbox{s.t.} \quad &\ts\sum_{e \in \pa_j} x_{e} \leq 1,   \qquad\quad \forall j \in [n] & \mbox{s.t.} \quad &b_e \ts\sum_{S \ni e} \alpha_S + \beta_{j(e)} \geq v_e, \qquad \forall e \in E \\
&\ts\sum_{e \in S} b_ex_{e} \leq f(S), \qquad \forall S \subseteq E & &\alpha_S, \beta_j \geq 0, \qquad\qquad\qquad\ \; \forall S \subseteq E,\, j \in [n]. \\
&x_e \geq 0,  \qquad\qquad\qquad\ \, \forall e \in E.
\end{align*}

In the dual, $j(e)$ denotes the index $j$ for which $e \in \pa_j$. The solution to the primal LP is the (fractional) optimal \emph{offline} solution to a given instance, and we use $\OPT$ to denote its value. By strong duality, the optimal values of the primal and dual LPs are the same. 


It will be useful for our analysis to re-write the dual objective in a different form. Let us define $\vec \gamma \in \Rp^E$ by $\gamma_e := \sum_{S \ni e} \alpha_S$. By a standard uncrossing argument, we may assume that the optimal dual $\vec \alpha$ is supported on a nested family of sets.   Thus, we can recover the part of the objective $\sum_{S \subseteq E} f(S) \cdot \alpha_S$ from $\vec 
 \gamma$ as
\begin{equation}
\label{eq:lovasz}
\sum_{S \subseteq E} f(S) \cdot \alpha_S = \int_0^\infty f(\{e \in E : \gamma_e \geq t\})\,dt.
\end{equation}
The right-hand-side integral is exactly the \emph{Lov\'asz extension} $L_f$ of a submodular function $f$. This is a natural continuous extension of the submodular function $f$ which has been studied in many contexts. Although equivalent formulations exist, it will be convenient for us to use the following definition. 

\begin{definition}\label{def:wl}
    Let $f$ be a monotone submodular function on $E$ with $f(\varnothing) = 0$. The \textit{Lov\'asz extension} $L_f : \Rp^E \to \R$ of $f$ is 
    \[
    L_f(\vec w) := \int_0^\infty f(\{e \in E : w_e \geq t\})\, dt.
    \]
\end{definition}
With this definition, we can write our dual in terms of $\gamma$ as follows.
\begin{align*}
    \min \quad &L_f(\vec \gamma) + \sum_{j = 1}^n \beta_j \\
    \text{s.t} \quad &b_e \gamma_e + \beta_{j(e)} \geq v_e, \qquad \text{for all $e \in E$} \\
    &\vec \gamma, \vec \beta \geq 0.
\end{align*}

\section{Water Level Machinery}
\label{sec:water-filling}

Recall that, we want to assign each element $e \in E$ a water level $w_e = w\up x_e$ which depends on the current allocation $\vec x$, where $\vec x$ satisfies the submodular constraints given by $f$:
\[ \vec{x}(S) \leq f(S) \quad \forall S \subseteq E. \]
Notice that, for the sake of defining water levels, we are assuming unweighted costs (i.e. $b_e = 1$). This removes clutter from our definitions, as we can simply add weights later by taking $\vec w\up {bx}$.

To understand how we define this water level vector $\vec w = \vec w\up x \in \Rp^E$ of allocation $\vec x$, we first enumerate several properties that the water levels should satisfy in order for the water-filling algorithm to work as it does for online bipartite matching: 
\begin{enumerate}
    \item (Monotonicity) $\vec w\up x$ is coordinate-wise non-decreasing in $\vec x$. 
    \item (Indication of Feasibility) $\vec w\up x \leq 1$ if and only if $\vec{x}(S) \leq f(S)$ for all $S \subseteq E$.
    \item (Locality) If $w\up x_{e_1} \neq w\up x_{e_2}$, then $\frac{\partial w\up x_{e_1}}{\partial x_{e_2}} = 0$.
    \item (Duality) $L_f(\vec w\up x) = \sum_{e \in E} x_e$.
\end{enumerate}

Monotonicity and feasibility indication are natural properties which intuitively require that $w_e$ is an indicator of how close $x_e$ is to being part of a tight constraint.
The need for locality and duality is less obvious, but they are important for the details of the primal-dual analysis of water-filling. Specifically, this is because the dual value $\gamma_e$ of an element $e \in E$ will be defined as a function of the water level $w_e$. The locality and duality properties are needed to relate increases in the dual objective term $L_f(\vec \gamma)$ to increases to the primal objective $\sum_{e \in E} x_e$ as the algorithm progresses.


\subsection{Definition of Water Levels and Equivalent Formulations}

In order to define a water level vector that satisfies our desired properties, it will be convenient (and enlightening) to provide both algorithmic and static definitions, which we will prove are equivalent. By juggling three different definitions, we are able to provide succinct proofs of the four key properties of water levels. 


First, let's consider a naive construction $\widetilde w\up x_e$. For $\vec x \in \Rp^E$, we define $\widetilde {\vec w}\up x_e = \max_{S \ni e} \frac{\vec x(S)}{f(S)}$, i.e. the maximum density
of a poly-matroid constraint involving $x_e$. Such a definition clearly satisfies monotonicity and indication of feasibility. However, this definition does not capture the water levels from the classic bipartite matching setting (i.e. in a partition matroid), and we can see this already in the simple setting of $E = \{1, 2\}$ and $f(S) = |S|$. In such a setting, the poly-matroid effectively only has the constraints $0 \leq x_1, x_2 \leq 1$, so we intuitively should let $w_e = x_e$. However, notice that if $x_1 < x_2$, then we have $\widetilde w_1 = \max_{S\ni 1} \frac{x(S)}{f(S)} = \frac{x_1 + x_2}{2}$. We see that this construction deviates from what we expect, and indeed we also find that our desired locality and duality properties are not satisfied. This problem arises because the heavier element $x_2$ influences the density of the densest constraint on $x_1$, despite the two variables being functionally independent.

The critical insight is that we can prevent this undesirable behavior by contracting sets with larger density before assigning the values of $w_e$ to sets with lower density. This inspires the following formulation of water levels in \Cref{def:water-levels}. 


\begin{definition}\label{def:water-levels}
The \emph{water level vector} $\vec w\upn {\vec x, f} \in \Rp^E$ (or $\vec w \up x$ when $f$ is known) with respect to an allocation $\vec x$ in $\Rp^E$ and a monotone submodular function $f : 2^E \to \Rp$, is defined as
\[ 
w_e \up x = w\upn {\vec x, f}_e := \max_{S \ni e} \min_{\substack{T \subseteq E \\ f_T(\{e\}) \neq 0}}   \frac{\vec x(S \setminus T)}{f_T(S)}
\]
where $f_T(S)$ denotes the contracted function $f(S \cup T) - f(T)$.
\end{definition}

Not only does this modification fix the problems with our naive definition, but it reveals a great degree of hidden structure. First, it happens to be the case that the min and max in \Cref{def:water-levels} are reversible, i.e. $w\up x_e = \min_T \max_S \frac{\vec{x}(S \setminus T)}{f_T(S)}$. We will prove this by showing the existence of a saddle point $S^*, T^*$: 
\begin{equation}
        \min_{\substack{T \subseteq E \\ f_T(\{e\}) \neq 0}} \frac{\vec{x}(S^* \setminus T)}{f_{T}(S^*)} = \frac{\vec{x}(S^* \setminus T^*)}{f_{T^*}(S^*)} = \max_{S \ni e} \frac{\vec{x}(S \setminus T^*)}{f_{T^*}(S)}.
    \end{equation}

Furthermore, these optimal sets form a nested family. We can find $\varnothing = S_0 \subseteq S_1 \subseteq \dots \subseteq S_k = E$ such that for each $e \in S_{\ell+1} \setminus S_\ell$, the pair $(S_{\ell + 1}, S_\ell)$ is a saddle point for $w_e\up x$ in the max-min expression of \Cref{def:water-levels}. We find the nested family via an intuitive and efficiently computable combinatorial description of water levels, \Cref{alg:comb-dec}.

 \begin{algorithm}[h]
  \caption{\texttt{A Combinatorial Presentation of Water Levels}}\label{alg:comb-dec}
  \Input{A point $\vec x \in \Rp^E$.}
  \BlankLine
  
  Initialize $\ell \leftarrow 0$, $S_0 \leftarrow \varnothing$.\;
  \While{$S_\ell \neq E$}{
    Let $S_{\ell+1}$ be the unique maximal set\footnotemark~in 
    $$\arg \max_{\substack{S \subseteq E \\ S \setminus S_{\ell} \neq \varnothing}} \frac{\vec{x}(S \setminus S_\ell)}{f_{S_\ell}(S)},$$
    i.e. the largest densest set over the contracted polymatroid.\; \label{alg-line:4}
    Let $t_{\ell+1} = \frac{\vec{x}(S_{\ell + 1} \setminus S_\ell)}{f_{S_\ell}
(S_{\ell+1})}$ be the density of $S_{\ell+1}$.\;
    Set $\alg w_e \leftarrow t_{\ell+1}$ for all $e \in S_{\ell+1} \setminus S_{\ell}$.\;
    $\ell \leftarrow \ell + 1$.
}
\Return{$\alg{\vec w}$.}
\end{algorithm}
\footnotetext{The maximal such set is unique due to the submodularity of $f$.}

\begin{restatable}[Saddle Point for Water Levels]{theorem}{levelsets}
  \label{lem:wl-alg-equiv}
    For any monotone submodular function $f: 2^E \to \Rp$ with $f(\varnothing) = 0$ and $x \in \Rp^E$, the vector $\vec w = \vec w\up x$ in \Cref{def:water-levels} has the property
    \[
    w\up x_e := \max_{S \ni e} \min_{\substack{T \subseteq E \\ f_T(\{e\}) \neq 0}}   \frac{\vec{x}(S \setminus T)}{f_T(S)} =  \min_{\substack{T \subseteq E \\ f_T(\{e\}) \neq 0}} \max_{S \ni e}  \frac{\vec{x}(S \setminus T)}{f_T(S)}.
    \]
    Moreover, the output $\alg{\vec w}$ of \Cref{alg:comb-dec} is equal to $\vec w$.
\end{restatable}

For readability, we delay the proof of the min-max property of \Cref{def:water-levels} and its equivalence to \Cref{alg:comb-dec} until \Cref{sec:wl-equiv}.

We also find an unexpected connection to market equilibria. Jain and Vazirani \cite{JV10} introduced a notion of \textit{submodular utility allocation markets} which can be described with the following convex program.
\begin{equation}\tag{SUA}\label{eq:SUA}
\begin{aligned}
    \max_u\quad &\ts \sum_{e \in E} m_e\log u_e \\
     \text{s.t.} \quad & \ts \sum_{e \in S} u_{e} \leq f(S), \quad \forall S \subseteq E\quad \qquad &&(\alpha_S) \\
    &u_{e} \geq 0. 
\end{aligned}
\end{equation}
It turns out the water levels of an allocation $\vec x$ 
can be computed from the optimal utilities of an SUA market where each element $e$ has weight $m_e := x_e$. \Cref{alg:comb-dec} gives us optimal duals to $\eqref{eq:SUA}$, which we also prove in \Cref{sec:wl-equiv}.

\begin{restatable}{theorem}{SUAthm}
  \label{thm:comb-dec-and-SUA}
Consider the vector $\alg{\vec w}$, the nested sets $S_1\subset \cdots \subset S_L$, and the levels $t_1,\ldots, t_L$
generated by \Cref{alg:comb-dec}. Then $t_1 > \cdots > t_L \geq 0$. Moreover, if we define 
\begin{align*}
    \alg \alpha_{S_L} &:= t_L \\
    \alg \alpha_{S_\ell} &:= t_{\ell} - t_{\ell + 1} \quad \ell = 1, \hdots, L - 1
\end{align*}
and $\alg \alpha_S = 0$ for all other $S \subseteq E$, then, $\alg u_e = \frac{x_e}{\alg w_e}$ is an optimal primal solution to \eqref{eq:SUA} and $\alg \alpha_S$ is an optimal dual solution.
\end{restatable}








\subsection{Key Properties of Water Levels}
Armed with the  characterizations of water levels, we show  they satisfy the desired properties.

\begin{proposition}[Monotone and Continuous]\label{cor:wl-mono}
    The vector $\vec w\up x$ is coordinate-wise non-decreasing in $\vec x$. Furthermore, $\vec w\up x$ is continuous with respect to $\vec x$. 
\end{proposition}
\begin{proof}
    Both properties follow immediately from \Cref{def:water-levels}, since the $w\up x_e$ is a maximum of minimums over monotone increasing linear functions on $\vec x$.
\end{proof}



\begin{proposition}[Indication of Feasibility]
\label{cor:wl-feasibility}
     $\vec w\up x \leq 1$ if and only if $\vec{x}(S) \leq f(S)$ for all $S \subseteq E$.
\end{proposition}
\begin{proof}
    This follows from \Cref{alg:comb-dec}. If $\vec w\up x \leq 1$, then in particular $t_1 = \max_{S \subseteq E}\frac{\vec{x}(S)}{f(S)} \leq 1$, which means $\vec{x}(S) \leq f(S)$ for all $S \subseteq E$. Conversely, if $\vec{x}(S) \leq f(S)$ for all $S \subseteq E$ then clearly $t_1 \leq 1$. Moreover, the densities $t_\ell$ are decreasing by \Cref{thm:comb-dec-and-SUA}, implying $t_\ell \leq 1$ for all $\ell$. Thus $w\up x \leq 1$. 
\end{proof}

\begin{proposition}[Duality]\label{lem:wl-duality}
    For any $\vec x \in \Rp^E$, we have
    $
    L_f(\vec w\up x) = \sum_{e \in E} x_e.
    $
\end{proposition}
\begin{proof}
For this, we refer to the convex program formulation of water levels from \Cref{thm:comb-dec-and-SUA}. Taking the optimal primal/dual pair $u_e, \alpha_S$ from \Cref{thm:comb-dec-and-SUA}, the complementary slackness conditions of \eqref{eq:SUA} give $\alpha_S \sum_{e \in S} u_e = \alpha_S f(S)$ for each $S \subseteq E$. Summing over those $S$ with $\alpha_S > 0$,
we get 
\[
    \sum_{\ell = 1}^L \alpha_{S_\ell} f(S_{\ell}) 
    \stackrel{(a)}{=} \sum_{\ell = 1}^L \alpha_{S_\ell} \sum_{e \in S_\ell} u_e
    \stackrel{(b)}{=} \sum_{\ell = 1}^L \alpha_{S_\ell} \sum_{e \in S_\ell} \frac{x_e}{w_e} = \sum_{e \in E} \left({\sum_{\ell: S_\ell \ni e} \alpha_{S_\ell}}\right) \cdot \frac{x_e}{w_e} \stackrel{(c)}{=} \sum_{e \in E} x_e.
\]
Here, $(a)$ is using  $\alpha_S \sum_{e \in S} u_e = \alpha_S f(S)$ for each $S \subseteq E$, $(b)$ is because $u_e = \frac{x_e}{w_e}$ by \Cref{thm:comb-dec-and-SUA}, and $(c)$ is because $\sum_{\ell: S_\ell \ni e} \alpha_{S_\ell} = w_e$, using the fact that the sets $S_\ell$ are nested and the definition of $\alpha_{S_\ell}$ in \Cref{thm:comb-dec-and-SUA}.
Finally, the LHS is equal to $L_f(w\up x)$ by \eqref{eq:lovasz} and using $\sum_{\ell: S_\ell \ni e} \alpha_{S_\ell} = w_e$.
\end{proof}

\begin{proposition}[Locality]\label{lem:wl-change}
    For $\vec x \in \Rp^E$ and any $e_1, e_2 \in E$ with $w\up x_{e_1} \neq w\up x_{e_2}$, we have $\frac{\partial w\up x_{e_1}}{\partial x_{e_2}} = 0$.
\end{proposition}
\begin{proof}
This follows from the continuity of $\vec w\up x$ with respect to $\vec x$ (implied by \Cref{def:water-levels}) and the algorithmic definition given by \Cref{alg:comb-dec}. Let $w\up x_{e_2} = t_\ell$, for some step $\ell$, where $t_1, \dots, t_L$ are the values from \Cref{alg:comb-dec} on input $\vec x$, and let $\eps := \frac{\min(t_{\ell-1} - t_\ell, t_\ell - t_{\ell + 1})}{2}$. By continuity, we may choose $\delta_\eps > 0$ small enough so that for any $\delta \in (-\delta_\eps, \delta_\eps)$ and $\vec y := \vec x + \delta \one_{e_2}$, we have
\[
    w_e\up{y} \in (w_e\up{x} - \eps, w_e\up{x} + \eps)
\]
for all $e \in E$. It suffices to show for any such $\vec y$ that $w\up {y}_{e_1} = w\up x_{e_1}$. Consider for each $t$ the sets $E\up{x}_{\geq t} := \{e : w_e\up{x} \geq t\}$ and $E\up{y}_{\geq t} := \{e : w_e\up{y} \geq t\}$. Then by our choice of $\eps$, we have 
\begin{align*}
    &E_+ := E\up{x}_{\geq t_\ell + \eps} = E\up{y}_{\geq t_\ell + \eps}, \text{ and}\\
    &E_- := E\up{x}_{\geq t_\ell - \eps} = E\up{y}_{\geq t_\ell - \eps}.
\end{align*}
Note that $e_2 \in E_- \setminus E_+$. 

It is clear then that the first $\ell - 1$ steps of the \Cref{alg:comb-dec} on $\vec x$ and on $\vec y$ are identical, with $S_{\ell-1} = E_{+}$ in both cases. This follows because $\vec x$ and $\vec y$ differ only in their value at $e_2$, and $w_{e_2}\up{y} < w_{e_2}\up{x} + \eps = t_\ell + \eps$, so the algorithm assigns water levels to all elements $e \in E_{+}$ before assigning a water level to $e_2$. Hence, if $e_1$ has $w\up x_{e_1} > w\up x_{e_2}$, we have $w\up {y}_{e_1} = w\up x_{e_1}$ as desired. 

In the other case that $e_1$ has $w\up{x}_{e_1} < w\up x_{e_2}$, observe that $e_1 \in E \setminus E_-$. Moreover, for both inputs $\vec x$ and $\vec y$, \Cref{alg:comb-dec} will have some step $\ell' \geq \ell$ ($\ell'$ may differ for $\vec x$ and $\vec y$) at which $S_{\ell'} = E_{-}$. Since $e_2 \in E_-$, and $\vec x$ and $\vec y$ differ only at $e_2$, then every future step of the algorithm is identical for the two inputs. In particular, $w\up {y}_{e_1} = w\up x_{e_1}$.
\end{proof}

In addition, a key consequence of the duality and locality properties is the following ``chain-rule'' lemma about the partial derivatives of $L_f(G(\vec w\up x))$ with respect to entries of $\vec x$. Such a lemma will be useful in the primal-dual analysis of online SAP.
\begin{lemma}[Water Level Chain Rule]\label{lem:L-G-derivative}
    If $G : \Rp \to \Rp$ is a non-decreasing differentiable function with continuous derivative $G' = g$, then for all $\vec x \in \Rp^E$ and $e \in E$,
    \[
    \frac{\partial (L_f(G(\vec w\up x)))}{\partial x_e} = g(w\up x_e),
    \]
    where $G(\vec{w}\up x) := (G(w\up x_e))_{e \in E}$.
\end{lemma}
\begin{proof}
    For given $\vec x \in \Rp^E$ and $e \in E$, let $\vec y = \vec x + \varepsilon \cdot \one_{e}$. Then expanding $L_f$ as an integral, applying a change of variables, and invoking the monotonicity of $G$, we have
    \begin{align*}
        L_f(G(\vec{w}\up y)) - L_f(G(\vec{w}\up x)) &= \int_0^\infty  \Big(f(\{e' : G(w\up y_{e'}) \geq t\}) - f(\{e' : G(w\up x_{e'}) \geq t\})\Big) dt \\
        &= \int_0^\infty  \Big(f(\{e' : G(w\up y_{e'}) \geq G(u)\}) - f(\{e' : G(w\up x_{e'}) \geq G(u)\})\Big) g(u) du \\
        &=\int_0^\infty  \Big(f(\{e' : w\up y_{e'} \geq u\}) - f(\{e' : w\up x_{e'} \geq u\})\Big) g(u) du.
    \end{align*}
Now, we crucially use locality (\Cref{lem:wl-change}) to claim that $\{e': w\up y_{e'} \geq u\} = \{e': w\up x_{e'} \geq u\}$ for any $u \not \in [w\up x_{e}, w\up y_{e}]$. This follows because for each $e' \in E$ and every $\vec z = \vec x + \delta \cdot \one_e$ for $\delta \in [0, \eps]$, we have $\frac{\partial w_{e'}\up{z}}{\partial z_e} = 0$ unless $w_{e'}\up{z} = w_e\up{z}$. In particular, if $w_{e'}\up{x}$ lies outside the range $[w\up x_{e}, w\up y_{e}]$, then the water level of $e'$ never changes as we move $\vec z$ from $\vec x$ to $\vec y$. Likewise, the water level of any $e'$ whose water level lies within the range $[w\up x_{e}, w\up y_{e}]$ may change, but it cannot increase beyond $w\up y_{e}$. So for each $e' \in E$, either $w\up y_{e'} = w\up x_{e'} \not \in [w\up x_{e}, w\up y_{e}]$ or $w\up x_{e'}, w\up y_{e'} \in [w\up x_{e}, w\up y_{e}]$.


Using this fact, we may restrict the above integral to the range $u \in [w\up x_e, w\up y_e]$. Together with the intermediate value theorem, we have that there exists some $\hat u \in [w\up x_e, w\up y_e]$ such that
\begin{align*}
    L_f(G(\vec{w}\up y)) - L_f(G(\vec{w}\up x)) &= \int_{w\up x_e}^{w\up y_e}  \Big(f(\{e' : w\up y_{e'} \geq u\}) - f(\{e' : w\up x_{e'} \geq u\})\Big) g(u) du \\
    &= g(\hat u) \cdot \int_{w\up x_e}^{w\up y_e}  \Big(f(\{e' : w\up y_{e'} \geq u\}) - f(\{e' : w\up x_{e'} \geq u\})\Big) du\\
    &=g(\hat u) \cdot (L_f(\vec{w}\up y) - L_f(\vec{w}\up x)).
\end{align*}
From duality (\Cref{lem:wl-duality}), we know that $L_f(\vec{w}\up y) - L_f(\vec{w}\up x) = \sum_{e' \in E} (y_{e'} - x_{e'}) = \epsilon$. Therefore, we ultimately get
\[
\frac{L_f(G(\vec{w}\up y)) - L_f(G(\vec{w}\up x))}{\epsilon} = g(\hat u).
\]
Taking the limit as $\epsilon \to 0$, we have $\hat u \to w\up x_e$ since $w\up x_e \leq \hat u \leq w\up y_e$. By continuity of $g$, we finally see that
\[
\frac{\partial (L_f(G(\vec{w}\up x)))}{\partial x_e} = \lim_{\epsilon \to 0} \frac{L_f(G(\vec{w}\up y)) - L_f(G(\vec{w}\up x))}{\epsilon} = g(w\up x_e). \qedhere
\]
\end{proof}

\subsection{Proofs of Equivalence of Water Level Definitions}\label{sec:wl-equiv}

In this sub-section, we prove the combinatorial decomposition in \Cref{alg:comb-dec} and the SUA market formulation both produce the water levels vector defined in \Cref{def:water-levels}.

\levelsets*

\begin{proof}
    To show the desired min-max property, it suffices to show that for any $e \in E$, there exist sets $S^*, T^*$ such that $e \in S$ and $f_T(\{e\}) \neq 0$ and
    \begin{equation}\label{eq:wl-min-max-prop}
        \min_{\substack{T \subseteq E \\ f_T(\{e\}) \neq 0}} \frac{\vec{x}(S^* \setminus T)}{f_{T}(S^*)} = \frac{\vec{x}(S^* \setminus T^*)}{f_{T^*}(S^*)} = \max_{S \ni e} \frac{\vec{x}(S \setminus T^*)}{f_{T^*}(S)}.
    \end{equation}
    From this, it follows that
    \begin{align*}
        \max_{S \ni e} \min_{\substack{T \subseteq E \\ f_T(\{e\}) \neq 0}}   \frac{\vec{x}(S \setminus T)}{f_T(S)} 
        &\geq \min_{\substack{T \subseteq E \\ f_T(\{e\}) \neq 0}} \frac{\vec{x}(S^* \setminus T)}{f_{T}(S^*)} \\
        &= \frac{\vec{x}(S^* \setminus T^*)}{f_{T^*}(S^*)} \\
        &= \max_{S \ni e} \frac{\vec{x}(S \setminus T^*)}{f_{T^*}(S)} \\
        &\geq \min_{\substack{T \subseteq E \\ f_T(\{e\}) \neq 0}} \max_{S \ni e}   \frac{\vec{x}(S \setminus T)}{f_T(S)}.
    \end{align*}
    Since we clearly have $\max_{S} \min_{T}   \frac{\vec{x}(S \setminus T)}{f_T(S)} \leq \min_{T} \max_{S}   \frac{\vec{x}(S \setminus T)}{f_T(S)}$ (omitting constraints on $S,T$) it follows that all inequalities above must be equalities.

    To show that such sets obeying \cref{eq:wl-min-max-prop} exist, we will show that they are exactly those given by the family $S_0, S_1, \dots, S_L$ resulting from \Cref{alg:comb-dec}. Specifically, we will show that for each $\ell$,
    \begin{equation}\label{eq:wl-alg-min-max}
        \min_{\substack{T \subseteq E \\ f_T(S_\ell) \neq 0}} \frac{\vec{x}(S_\ell \setminus T)}{f_{T}(S_\ell)} = \frac{\vec{x}(S_\ell \setminus S_{\ell - 1})}{f_{S_{\ell - 1}}(S_\ell)} = \max_{S \setminus S_{\ell-1} \neq \varnothing} \frac{\vec{x}(S \setminus S_{\ell - 1})}{f_{S_{\ell - 1}}(S)}.
    \end{equation}

    From this, it follows that for each $e \in S_\ell \setminus S_{\ell - 1}$, the pair of sets $(S_{\ell - 1}, S_{\ell})$ form an optimal pair $(S^*, T^*)$ in \cref{eq:wl-min-max-prop} for $e$. This not only completes the proof of our min-max property, but also implies $w_e = \frac{\vec{x}(S_\ell \setminus S_{\ell - 1})}{f_{S_{\ell - 1}}(S_\ell)} = t_\ell = \alg w_e$.

    We proceed will show by induction that \cref{eq:wl-alg-min-max} holds for each $\ell \geq 1$. Notice that the second equality in \cref{eq:wl-alg-min-max} holds by choice of $S_\ell$ in \Cref{alg:comb-dec}, so we only need to show the first equality.

    For the case $\ell = 1$, suppose some set $T$ has $\frac{\vec{x}(S_1 \setminus T)}{f_T(S_1)} < t_1$. We may assume $T \subseteq S_1$, since otherwise we can take $T \cap S_1$ since $f_T(S_1) \leq f_{T \cap S_1}(S_1)$. Then we have
    \[
    t_1 = \frac{\vec{x}(S_1)}{f(S_1)} = \frac{\vec{x}(S_1 \setminus T) + \vec{x}(T)}{f_T(S_1) + f(T)}.
    \]
    Since we know $\frac{\vec{x}(S_1 \setminus T)}{f_T(S_1)} < t_1$, it must then be true that $\frac{\vec{x}(T)}{f(T)} > t_1$. However, this is impossible as $S_1$ is chosen to have maximum density.

    For the inductive step, let $\ell \geq 1$ and  assume \cref{eq:wl-alg-min-max} holds for all previous steps $\ell' < \ell$. Our reasoning will be similar to the base case, but with some extra steps.  As before, take some $T \subseteq E$ with minimum possible value of $\frac{\vec{x}(S_\ell \setminus T)}{f_T(S_\ell)}$, and suppose $\frac{\vec{x}(S_\ell \setminus T)}{f_T(S_\ell)} < t_\ell$. We may again assume $T \subseteq S_\ell$. In addition, we assume that there are no $e' \not \in T$ with $f_T(\{e'\}) = 0$, since adding such $e'$ to $T$ can only improve the choice of $T$.

    We consider two cases: either $S_{\ell - 1} \subseteq T$, or $S_{\ell-1} \setminus T \neq \varnothing$. In the former case we have
    $$
    t_\ell = \frac{\vec{x}(S_\ell \setminus S_{\ell-1})}{f_{S_{\ell-1}}(S_{\ell})} 
    = \frac{\vec{x}(S_\ell \setminus T) + \vec{x}(T \setminus S_{\ell-1})}{f_{T}(S_{\ell}) + f_{S_{\ell-1}}(T)}.
    $$
    Since we assumed that $\frac{\vec{x}(S_\ell \setminus T)}{f_{T}(S_{\ell})} < t_\ell$, it must be the case that $\frac{\vec{x}(T \setminus S_{\ell-1})}{f_{S_{\ell-1}}(T)} > t_\ell$. However, this contradicts the choice of $t_\ell$ in \Cref{alg:comb-dec} as the maximum density of a set after contracting $S_{\ell - 1}$.

    Now, consider the second case where $S_{\ell - 1} \setminus T$ is nonempty. Pick $e' \in S_{\ell - 1} \setminus T$, and let $\ell' < \ell$ be such that $\alg w_{e'} = t_{\ell'}$. We then have

    $$
    \frac{\vec{x}(S_{\ell} \setminus T)}{f_{T}(S_\ell)} 
    = \frac{\vec{x}(S_\ell \setminus (S_{\ell'} \cup T)) + \vec{x}(S_{\ell'} \setminus T)}{f_{S_{\ell'} \cup T}(S_{\ell}) + f_{T}(S_{\ell'})}.
    $$

    Notice that $\frac{\vec{x}(S_{\ell'} \setminus T)}{f_{T}(S_{\ell'})} \geq \frac{\vec{x}(S_{\ell'} \setminus S_{\ell'-1})}{f_{S_{\ell'-1}}(S_{\ell'})} = t_{\ell'}$ by \cref{eq:wl-alg-min-max} applied to $\ell'$, which holds by our inductive hypothesis. Since $\frac{\vec{x}(S_{\ell'} \setminus T)}{f_{T}(S_{\ell'})} \geq t_{\ell'} > t_\ell > \frac{\vec{x}(S_{\ell} \setminus T)}{f_{T}(S_\ell)}$, this means that we must have $\frac{\vec{x}(S_\ell \setminus (S_{\ell'} \cup T))}{f_{S_{\ell'} \cup T}(S_{\ell})} < \frac{\vec{x}(S_{\ell} \setminus T)}{f_{T}(S_\ell)}$. However, this contradicts the choice of $T$.

\end{proof}

\SUAthm*
\begin{proof}
    We begin with an rephrasing of \Cref{alg:comb-dec}. Instead of contracting elements as we did in \Cref{alg:comb-dec}, we ``freeze" elements in \Cref{alg:alt-comb-dec}. Scaling $\vec x$ until some set is saturated is equivalent to measuring the multiplicative slack of sets; therefore, the densities in \Cref{alg-line:4} of \Cref{alg:comb-dec} are precisely the time steps at which we freeze a new set of elements in \Cref{alg:frozen} of \Cref{alg:alt-comb-dec}.
\begin{algorithm}[h]
   \caption{\texttt{An Alternate Combinatorial Presentation of Water Levels}}\label{alg:alt-comb-dec}
   \Input{A point $\vec x \in \Rp^n$.}
   Initialize $t = 0$, all elements are considered ``unfrozen"\;
   \While{there exists an unfrozen element}{
     Raise $t$ until the vector  
     \[ \vec x^{(t)} = \begin{cases}
         t \cdot x_e & \text{if $e$ is unfrozen} \\
         t_{\text{frozen}}(e) \cdot x_e & \text{if $e$ is frozen}
     \end{cases} \]
     has a tight set including at least one unfrozen element.\; 
     Freeze all the elements in the (unique) largest such tight set $S_t$ of $t \cdot \vec x$. \label{alg:frozen} \;
     We set all newly frozen elements in $S$ to have $t_{\text{frozen}}(e) := t$.\;
     Set $\alg w_e = \frac{1}{t}$ for all $e \in S_t$.
 }
 \Output{A vector $\alg{\vec w} $.}
 \end{algorithm}
    
    We use the KKT conditions to show $\alg u_e$ and $\alg \alpha_S$ are optimal. Denote the Lagrange multipliers for the constraints $\alg u_e \geq 0$ by $\mu_e$. We will set $\mu_e = 0$ if $x_e > 0$ and $\mu_e = w_e$ otherwise. The KKT conditions are as follows:
    \begin{itemize}
        \item \emph{Primal Feasibility:} $\sum_{e \in S} \alg u_e \leq f(S)$ for all $S \subseteq E$.
        \item \emph{Dual Feasibility:} $\alg \alpha_S \geq 0$ for all $S \subseteq E$. 
        \item \emph{Stationarity Conditions:} For all $e \in E$, 
        \[ \frac{x_e}{\alg u_e} = \sum_{S \ni e} \alg \alpha_S - \mu_e. \]
        \item \emph{Complementary Slackness:} $\alg \alpha_S > 0$ implies $\sum_{e \in S} \alg u_e = f(S)$ and $\mu_e \cdot \alg u_e = 0$.
    \end{itemize}
    Primal feasibility follows from the fact that \Cref{alg:alt-comb-dec} maintains feasibility of $t \cdot \vec x$ on unfrozen elements. Dual feasibility also easily follows since $t$ only rises. If $x_e > 0$, then $\mu_e = 0$ and 
    \[  
        \frac{x_e}{\alg u_e} = \alg w_e = \sum_{S \ni e} \alg \alpha_S
    \]
    as desired. Otherwise, if $x_e = 0$, then since $\mu_e = \alg w_e$, the stationary condition still holds. Lastly, we check complementary slackness. We have a positive $\alg \alpha_S$ precisely on the sets $E_1, \hdots, E_L$, and so it suffices to check these sets are tight. Indeed, by definition of \Cref{alg:alt-comb-dec}, these sets are tight. Lastly, we have that if $x_e > 0$, then $\mu_e = 0$ and that finishes our complementary slackness conditions. 
\end{proof}

\section{A Warm Up: Online Matroid Intersection}
\label{sec:fractional}
With our new vector of water levels $\vec w\up x$ defined, we will see how they can be used to naturally extend the water-filling paradigm.
Before proving \Cref{thm:online-SAP} for general online SAP, we will see how our techniques can be applied in a simpler setting: fractional online matroid intersection with part arrival. 

In \textit{Online Matroid Intersection}, the goal is to maximize the size of a common independent set between two matroids, when the elements are initially unknown and arrive in some online fashion. We focus on the case where one of the matroids is a partition matroid. In particular, suppose $\mathcal{M}$ is an arbitrary matroid and $\parm$ is a partition matroid with parts $\pa_1, \hdots, \pa_n$, both defined over (an initially unknown) ground set $E$.  We have access to an independence oracle for $\mathcal{M}$ restricted to the elements which have been revealed so far; in other words, $\mathcal{M}$ is known offline. Parts from the partition matroid arrive online. When a part $\pa_j$ arrives, the elements in $\pa_j$ are revealed, and we immediately and irrevocably choose at most one element from $\pa_j$. The goal is to maximize the cardinality of the set of chosen elements, subject to the set being independent in both matroids. In the fractional version of this problem, instead of choosing one $e \in \pa_j$, we select values $x_e \geq 0$ for $e \in \pa_j$ so that $\vec x(\pa_j) \leq 1$ and $x$ remains feasible in the matroid polytope. Online matroid intersection, and the corresponding fractional problem, are instances of online (fractional) SAP where $v_e = b_e = 1$ for all $e \in E$, and $f(S) := \rank_{\mathcal{M}}(S)$.

Our fractional algorithm for this problem will, upon receiving part $\pa_j$, continuously allocate infinitesimally small $dx_e$ to $x_e$ for some $e \in \pa_j$ which has minimum water level, i.e. $e \in \argmin_{e' \in \pa_j} w\up x_{e'}$. This continues until either 1 unit of water has been output by $\pa_j$ (the constraint $\vec x(\pa_j) = 1$ becomes tight) or no more water can be output because every $e \in \pa_j$ has water level $w\up x_e = 1$ (each such $e$ is part of a tight constraint $\vec x(S) = f(S)$). The algorithm then moves onto the next arriving part and repeats the process. We use a primal-dual analysis to prove that the algorithm is $(1-\nicefrac{1}{e})$-competitive.

\subsection{Analysis}

We proceed with a primal-dual analysis to prove \Cref{thm:online-SAP} for fractional online matroid intersection.  The dual program is:
\begin{align*}
    \min \quad &L_f(\vec \gamma) + \sum_{j = 1}^n \beta_j \\
    \text{s.t} \quad &\gamma_e + \beta_{j(e)} \geq 1, \qquad \text{for all $e \in E$} \\
    &\vec \gamma, \vec \beta \geq 0.
\end{align*}
where $L_f$ is the Lov\'asz extension of $f$, and $j(e)$ is defined such that $e \in \pa_{j(e)}$. We will construct a set of dual variables based on the primal allocation. Specifically, upon each arrival of $\pa_j$, we start by setting $\beta_j = 0$. Then upon each infinitesimal increase of $x_e$ by $dx_e$ for $e \in \pa_j$, we increase $\vec \gamma$ to maintain
$$
\gamma_{e'} := G(w_{e'}) \quad \text{for all $e' \in E$}
$$
and increase $\beta_j$ by
$$
d\beta_j := (1 - g(w_{e}))dx_e.
$$
where $g(x) := e^{x-1}$, and $G(x) := \int_0^x g(t)\,dt = e^{x-1}-e^{-1}$. Observe that since the algorithm only increases the primal allocation $x$, by \Cref{cor:wl-mono} the dual variables also only increase as the algorithm progresses. 

To show a $1 - \nicefrac{1}{e}$ competitive ratio, we need to show $(1 - \nicefrac{1}{e})$-approximate feasibility of the dual, and that the dual increase is at most the primal increase. 

\subsubsection{Approximate Feasibility}
We will show that, immediately after the allocation to $\pa_j$ completes, we have $\gamma_e + \beta_{j} \geq 1-1/e$ for all $e \in \pa_j$. Since dual values only increase as the algorithm progresses, this would imply inequality also holds for the final dual values.

Let $w^* = \min_{e \in \pa_j} w\up x_{e}$ be the minimum water level of an element of $\pa_j$ immediately following the allocation to $\pa_j$. We claim that $\beta_j \geq 1 - g(w^*)$. To see this, notice that $d\beta_j \geq (1 - g(w^*))dx_e$ at each point in time during the allocation, so we clearly have $\beta_j \geq (1 - g(w^*))\sum_{e \in \pa_j} x_e$. If $\sum_{e \in \pa_j} x_e = 1$, we have our claim. Otherwise, every $e \in \pa_j$ must be involved in a tight offline constraint, so $w^* = 1$. In this case, $1 - g(w^*) = 0 \leq \beta_j$.

Using this claim, we easily obtain for each element $e \in \pa_j$,
\[
    \gamma_e + \beta_j \geq G(w_e) + 1 - g(w^*) \geq G(w_e) + 1 - g(w_e) = 1 - 1/e.
\]

\subsubsection{Primal equals Dual}
To show that the primal objective equals the dual objective, we will show that the rate of change in primal and dual objectives are equal at each instant in the continuous allocation. In the allocation to $\pa_j$, when $x_e$ receives infinitesimal allocation $dx_e$, the change in the primal objective is exactly $dx_e$.

Meanwhile, the change in dual objective is $d(L_f(G(w\up x))) + d\beta_j$. By \Cref{lem:L-G-derivative}, we know that $d(L_f(G(w\up x))) = g(w_e)dx_e$, and by definition of $\beta_j$ we have $d\beta_j = (1 - g(w_e))dx_e$. Hence, we have that the change in dual is also $g(w_e)dx_e + (1 - g(w_e))dx_e = dx_e$.

\subsubsection{Proof of the Main Theorem}\label{subsec:weak-duality}

The above discussion immediately gives the proof of the main theorem:

\begin{proof}[Proof of \Cref{thm:online-SAP} for online matroid intersection]
Let $x$ be the primal allocation given by water-filling, and $\vec \gamma, \vec \beta$ the associated dual assignment defined above. Since we showed that at each step of the algorithm, $\Delta\text{Primal} \geq \Delta\text{Dual}$, we have that $\sum_e x_e \geq L_f(\vec \gamma) + \sum_j \beta_j$. By approximate feasibility, we have that $\vec \gamma' := \frac{e}{e-1} \cdot \vec \gamma$ and $\vec \beta' := \frac{e}{e-1} \cdot \vec \beta$ together form a feasible dual solution. Finally, positive homogeneity\footnote{$L_f(\lambda x) = \lambda L_f(x)$ for $\lambda \geq 0$. This follows from the definition of Lov\'asz extension.} of the Lov\'asz extension and duality together give
\begin{align*}
    \sum_e x_e &\geq L_f(\vec \gamma) + \sum_j \beta_j \\
    &= \left(1 - \frac{1}{e}\right) \cdot \bigg(L_f(\vec \gamma') + \sum_j \beta_j'\bigg) \\
    &\geq \left(1 - \frac{1}{e}\right) \cdot \OPT_{\dual} = \left(1 - \frac{1}{e}\right) \cdot \OPT. \qedhere
\end{align*}
\end{proof}

\section{General Online Submodular Assignment}
\label{sec:gap}



In order to motivate our algorithm for online SAP, it will be useful to adopt an economics perspective similar to that presented in \cite{birnbaum2008line}. We think of each online arrival $j$ as a bidder and each offline element $e \in E$ as a good. Each bidder $j$ desires at most one unit of goods from the set $\pa_j$, and receives utility $v_e$ for each unit of good $e$ received. We, as the seller, allocate a quantity $x_e$ of good $e$ to bidder $j$ for each $e \in \pa_j$, but we are also limited the supply constraints $\vec{bx}(S) \leq f(S)$ for each $S \subseteq E$. Our goal is to maximize the welfare $\sum_{e \in E} v_e x_e$.

In this setting, it will be important to discriminate between items based on their ``bang-per-buck'' $\frac{v_e}{b_e}$. This is due to the fact that, with free disposal, it can be beneficial to discard an allocation with small bang-per-buck in order to make space for a higher value item. To this end, for $t \geq 0$ define the vector $\vec{w}^t$ to be
\[
\vec{w}^t := \vec w\upn{(b_e x_e)_{e \colon v_e/b_e \geq t}},
\]
i.e. the water levels of items when only considering allocations on item $e$ with bang-per-buck at least $t$.

Using this definition, we give a pricing-based algorithm for allocating $x_e$ values. For each $e \in E$, we place an instantaneous per-unit price on $e$ of
\[ p_e := b_e \cdot \int_0^{v_e/b_e} g(w^t_{e})dt. \]

Then, upon the arrival of $\pa_j$, bidder $j$ will continuously ``purchase'' an infinitesimal amount $dx_e$ of good $e$, for some $e \in \pa_j$ which has highest marginal utility, i.e. $e \in \argmax_{e' \in \pa_j} (v_{e'} - p_{e'})$. If $e$ is not part of any tight set, then we may increase it freely. Otherwise, in order to accommodate the new increase to $x_e$, we will decrease the allocation in $\ealt$ by $dx_\ealt= -\frac{b_e}{b_{\ealt}} \cdot dx_e$, where
\[ 
    \ealt \in \argmin_{e'} \Bigg\{ \frac{v_{e'}}{b_{e'}} \, \colon \, e' \in \bigcap_{\substack{S \ni e \\ \vec{bx}(S) = f(S)}} S \Bigg\}.
\]

In other words, $\ealt$ is the lowest bang-per-buck element that, upon decreasing $x_\ealt$, creates space for an increase in $x_e$. Note that the choice of $\ealt$ is always non-empty as it contains $e$. This continues until either 1 unit of good is allocated to $\pa_j$ (the constraint $x(\pa_j) = 1$ becomes tight) or no good produces positive utility ($v_e = p_e$ for each $e \in \pa_j$). The algorithm then moves onto the next arriving part and repeats the process. 

\subsection{Analysis}\label{sec:frac-osap}

We proceed with a primal-dual analysis to prove \Cref{thm:online-SAP}. Recall the dual program from \Cref{sec:prelims}:
\begin{align*}
    \min \quad &L_f(\vec{\gamma}) + \sum_{j = 1}^n \beta_j \\
    \text{s.t} \quad &b_e \cdot \gamma_e + \beta_{j(e)} \geq v_e, \qquad \text{for all $e \in E$} \\
    &\vec \gamma, \vec \beta \geq 0.
\end{align*}
Upon each arrival of $\pa_j$, we start by setting $\beta_j = 0$. Then upon each infinitesimal increase of $x_e$ by $dx_e$ for $e \in \pa_j$, we maintain
$$
\gamma_{e'} := \int_0^\infty G(w^t_{e'})\,dt \quad \text{for all $e' \in E$}
$$
and increase $\beta_j$ by
$$
d\beta_j := (v_e - p_e)dx_e = b_e \left(\int_0^{v_e/b_e}(1 - g(w^t_{e}))dt \right)dx_e.
$$
where $g(x) := e^{x-1}$, and $G(x) := \int_0^x g(t)\,dt = e^{x-1}-e^{-1}$.  


\subsubsection{Approximate Feasibility}\label{sec:frac-osap-af}
The goal is to show that immediately after the allocation to $\pa_j$ completes, we have $b_e \gamma_e + \beta_{j} \geq (1-\nicefrac{1}{e})v_e$ for all $e \in \pa_j$. Again, by monotonicity of water levels (\Cref{cor:wl-mono}) the dual variables also only increase as the algorithm progresses, so this implies the approximate feasibility for the final dual values.

We denote $u^* = \max_{e \in \pa_j} (v_e - p_e)$ to be the maximum marginal utility of an element of $\pa_j$ immediately following the allocation to $\pa_j$. We claim that $\beta_j \geq u^*$. At every point during the allocation, $d\beta_j \geq u^* dx_e$, so clearly $\beta_j \geq u^*\sum_{e \in \pa_j} x_e$. If $\sum_{e \in \pa_j} x_e = 1$ and we have our claim. Otherwise, we must have reached zero marginal utility during our allocation, so $u^* = 0$, in which case $\beta_j \geq u^*$ trivially.

Using this bound on $\beta_j$, we conclude
\[
    b_e \gamma_e + \beta_j \geq b_e \int_0^\infty G(w^t_e)\,dt + u^* \geq b_e \int_0^{v_e/b_e} \left( G(w^t_e) + 1 - g(w^t_e)  \right)dt= v_e(1 - \nicefrac{1}{e}).
\]
for each element $e \in \pa_j$.

\subsubsection{Primal exceeds Dual}\label{sec:frac-osap-pd}
We will show the primal exceeds the dual by instead dealing with a ``dual surrogate" which upper bounds the dual value.

\[L_f(\vec{\gamma}) = L_f\left( \int_0^\infty G(\vec{w}^t) \, dt \right) \leq \int_0^\infty L_f(G(\vec{w}^t)) \, dt \]
The left hand term is precisely the $\vec{\gamma}$ term in the dual objective. We call the right hand term the $\vec{\gamma}$ term of the ``dual surrogate". The inequality follows from Jensen's inequality and the convexity of $L_f$. 

We will show that the rate of change in primal and dual surrogate objectives is equal at each instant in the continuous allocation. In the allocation to $\pa_j$, when $x_e$ receives infinitesimal allocation $dx_e$, note that $x_{\ealt}$ is infinitesimally reduced by $\frac{b_e}{b_{\ealt}} \cdot dx_e$. Hence, the change in the primal objective is $b_e (\nicefrac{v_e}{b_e} - \nicefrac{v_{\ealt}}{b_{\ealt}})dx_e$.

Meanwhile, the change in surrogate dual objective is $\int_0^{\infty} \left(dL_f(G(\vec w^t)\right) dt + d\beta_j$. Using the chain rule lemma for water levels (\Cref{lem:L-G-derivative}) with the fact that $\vec w^t = \vec w\up{(b_e x_e)_{e : v_e/b_e \geq t}}$, we know that 
\[d(L_f(G(\vec w^t)) =
\begin{cases}
    g(w^t_e) b_e \cdot dx_e & v_{\ealt}/b_{\ealt} \leq t \leq v_e/b_e\\
    0 & \text{otherwise}.
\end{cases}
\]
To see why $d(L_f(G(\vec w^t))$ is only non-zero for $t \in [v_{\ealt}/b_{\ealt}, v_e/b_e]$, observe that if $t > v_e/b_e$, then clearly $x_e$'s value does not effect $\vec{w}^t$. On the other hand, if $t < v_{\ealt}/b_{\ealt}$, we claim the water level of $x_e$ remains 1 while deallocating $\ealt$ and allocating $e$. Recall $\ealt$ was chosen as the minimum bang-per-buck element in $S_{\text{tight}} = \bigcap_{\text{tight $S \ni e$}} S$. Note that $S_{\text{tight}}$ is itself a tight set, one that remains tight if we restrict our allocation to elements of bang-per-buck at least $t$. $S_{\text{tight}}$ also remains tight while shifting mass from $\ealt$ to $e$. Thus, the water level vector $\vec{w}^t$ remains identical.

By definition of $\beta_j$ we have $d\beta_j = b_e \left(\int_0^{v_e/b_e}(1 - g(w^t_{e}))dt \right)dx_e$. Hence, in total, we have that the change in surrogate dual is 
\[
\int_0^{\infty} \left(dL_f(G(\vec w^t)\right) dt + d\beta_j
= b_e \left(\int_{v_{\ealt}/b_{\ealt}}^{v_e/b_e} g(w^t_e) dt\right) dx_e + b_e \left(\int_0^{v_e/b_e}(1 - g(w^t_{e}))dt \right)dx_e.
\]
Notice that $w^t_e = 1$ when $t \leq \nicefrac{v_{\ealt}}{b_{\ealt}}$, so we have $(1 - g(w^t_{e})) = 0$ for such $t$. Using this, we can simplify the change in surrogate dual as 
\[
b_e \left(\int_{v_{\ealt}/b_{\ealt}}^{v_e/b_e} g(w^t_e) dt\right) dx_e + b_e \left(\int_{v_{\ealt}/b_{\ealt}}^{v_e/b_e}(1 - g(w^t_{e}))dt \right)dx_e = b_e \left(\frac{v_e}{b_e} - \frac{v_\ealt}{b_\ealt}\right) dx_e.
\]

The above discussion applied to an argument identical to that of \Cref{subsec:weak-duality} immediately gives the proof of the main theorem.

\subsection{Integral Algorithm under Small Bids Assumption}\label{sec:small-bids}

So far, we analyzed a fractional algorithm for online SAP and showed that it gets a $1-\nicefrac{1}{e}$ competitive ratio. Now, we will show that a similar algorithm and analysis applies to integral SAP under a small bids assumption. To illustrate the key ideas, in this section we will only examine the AdWords version of SAP (in other words, when $b_e = v_e$ for all $e \in E$). The AdWords version of SAP does not require free disposal, which makes the analysis simpler.

Recall the small bids assumption for SAP:
\smallbids*

The small bids assumption allows us to prove the following lemma, which essentially states that the water level is an $\epsilon$-Lipschitz function of the allocation. Intuitively, this is useful because it means that selecting any single element can increase the water levels by at most $\epsilon$, which allows for an analysis that approximates the fractional case. Indeed, when we analyze the integral algorithm under the small bids assumption, we will only use \Cref{lem:wl-lipschitz}, and not use \Cref{ass:small_bids} directly.
\begin{lemma}[Water Levels are Lipschitz]
    \label{lem:wl-lipschitz}
    Suppose \Cref{ass:small_bids} holds. Let $x \in \Rp^E$ and suppose $y = x + t \one_{e}$ for some $t \geq 0$ and $e \in E$. Then we have
    $$\norm{\vec w\up{bx} - \vec w\up{by}}_{\infty} \leq \epsilon t.$$
\end{lemma}

\begin{proof}
    Note that since $\vec y \geq \vec x$, we have $\vec{w} \up{by} \geq \vec{w} \up{bx}$ by monotonicity (\Cref{cor:wl-mono}). Furthermore, by locality (\Cref{lem:wl-change}), since $y$ is obtained from $x$ by increasing the allocation on a single coordinate $e$, the element whose water level increased the most is that of $e$ itself. Thus it suffices to show that $w_e\up{by} \leq w_e\up{bx} + \epsilon t$. 
    From the definition of water levels (\Cref{def:water-levels}), we have 
    $$
    w\up{bx}_e := \max_{S \ni e} \min_{\substack{T \subseteq E \\ f_T(\{e\}) \neq 0}}   \frac{\vec{bx}(S \setminus T)}{f_T(S)}.
    $$
    Let $S_x$ and $T_x$ be the sets that attain the above $\max \min$ for $w_e\up{bx}$. Analogously, define $S_y$ and $T_y$ to be the sets that attain the $\max\min$ for $w_e\up{by}$. Then
    \begin{align*}
        w_e\up{by} 
        &= \frac{\vec{by}(S_y \setminus T_y)}{f_{T_y}(S_y)} 
        \leq \frac{\vec{by}(S_y \setminus T_x)}{f_{T_x}(S_y)} \\
        &= \frac{\vec{bx}(S_y \setminus T_x)}{f_{T_x}(S_y)}
        + \frac{\vec{b}(\vec{y} - \vec{x})(S_y \setminus T_x)}{f_{T_x}(S_y)} \\
        &\leq \frac{\vec{bx}(S_x \setminus T_x)}{f_{T_x}(S_x)}
        + \frac{\vec{b}(\vec{y} - \vec{x})(S_y \setminus T_x)}{f_{T_x}(S_y)} \\
        &= w_e\up{bx} + \frac{\vec{b}(\vec{y} - \vec{x})(S_y \setminus T_x)}{f_{T_x}(S_y)}.
    \end{align*}
    Since $\vec y - \vec x = t\one_{e}$, we have
    $$\frac{\vec{b}(\vec{y} - \vec{x})(S_y \setminus T_x)}{f_{T_x}(S_y)}
    \leq  \frac{tb_e}{f_{T_x}(S_y)}
    \leq \frac{tb_e}{f_{T_x}(\{e\})} 
    \leq t\epsilon. 
    $$  
    Here, the second inequality is because $e \in S_y$ and $f$ is monotone, and the last inequality is by the small-bids assumption. 
\end{proof}

With \Cref{lem:wl-lipschitz} in hand, we are now ready to analyze integral online SAP under the small-bids assumption.

\thmsmallbids*

\begin{proof}[Proof of \Cref{thm:small-bids} 
.]
    
    Define $f'(S) := (1-\epsilon) f(S)$. Since $f$ is monotone and submodular, so is $f'$. \emph{In this proof, all water levels will be with respect to $f'$, unless explicitly noted otherwise.} 
    
    Consider the following integral algorithm:
    \begin{enumerate}
        \item Initialize $x = 0$. ($x$ is  the integral allocation to be returned.)
        \item Upon the arrival of part $j$, pick
        $$e_j \in \argmax\left\{b_e\left(1 - g(w_e\up{bx}) \right): e \in Q_j, \, w_e\up{bx} < 1 \right\}.$$
        \item Update $x \gets x + \one_{e_j}$. (If there is no $e \in Q_j$ with $w_e\up{bx} < 1$, leave $x$ unchanged.)
    \end{enumerate}
    The idea here is the same as the fractional algorithm: always select the item with the highest utility. The only difference is that we are artificially scaling down the capacity constraints by a factor of $1 - \epsilon$. The reason for this is to ensure that the allocation returned by the integral algorithm is always feasible to the original problem. Indeed, since the integral algorithm only allocates items whose water level (with respect to $f'$) is less than 1, \Cref{lem:wl-lipschitz} and locality  (\Cref{lem:wl-change}) imply that the final water levels (again, under $f'$), are at most $1 + \epsilon$. Since $f' = (1-\epsilon)f$, this implies that the final water levels of the allocation with respect to the original function $f$ are at most $(1+\epsilon)(1-\epsilon) < 1$, which by indication of feasibility (\Cref{cor:wl-feasibility}) means that the integral algorithm is guaranteed to produce a feasible solution. 

    Having seen that the solution returned by the integral algorithm is always feasible, we now proceed to bound its competitive ratio. We recall the primal and dual LPs (with capacity constraints given by $f'$) below:
    \begin{align*}
    \max \quad &\sum_e b_ex_e & \min \quad &L_{f'}(\gamma) + \sum_j \beta_j \\
    \text{s.t.} \quad &\sum_{e \in Q_j} x_e \leq 1 \qquad \forall \; j & \text{s.t.} \quad &b_e\gamma_e + \beta_{j(e)} \geq b_e \qquad \forall e \in E\\
    &\sum_{e \in S} b_ex_e \leq f'(S) \qquad \forall \; S \subseteq E & &\gamma, \beta \geq 0.\\
    &x_e \geq 0 \qquad \forall e \in E
\end{align*}
\textbf{Setting the dual variables.} Consider the arrival of part $j$. Let $x$ denote the allocation \emph{right before} $j$ arrives, and let $x'$ denote the allocation \emph{right after} $j$ arrives. After $j$ makes its allocation, we update the dual variables as follows: 
\begin{itemize}
    \item $\gamma_e = G(w_e\upn{\vec{bz}'})$ for all $e \in E$.  
    \item $\beta_j = \max\{0, \max \{b_e(1 - g(w_e\up{bx})): e \in Q_j \}\}$. 
\end{itemize}
Note that the values of $\gamma_e$ are updated in each iteration, whereas the values of $\beta_j$ are updated once in iteration $j$ and then remain unchanged forever. Moreover, note that the dual variables are non-decreasing throughout the course of the algorithm, because the water levels are non-decreasing and $g$ is an increasing function. 
\paragraph{Approximate dual feasibility.} Consider any $e \in E$, and let $j = j(e)$. Let $x$ denote the \emph{final allocation} of the algorithm. Then,
$$b_e\gamma_e + \beta_j = b_eG(w_e\up{bx}) + \beta_j \stackrel{(a)}{\geq} b_eG(w_e\up{bx}) + b_e(1 - g(w_e\up{bx})) \stackrel{(b)}{=} \left(1 - \frac{1}{e}\right) b_e.$$
Here, (a) uses the fact that the water levels are non-decreasing throughout the course of the algorithm, and $g$ is increasing. (b) is using the definition of $g(t) = e^{t-1}$.  

\paragraph{Change in primal vs. change in dual.} The last thing we need to show is that the value of the algorithm is not too small compared to the value of the dual solution. To do this, we compare the change in the primal to the change in the dual in each iteration. 

Consider the arrival of part $j$. Let $\vec $ denote the algorithm's allocation right before $j$ arrives. There are two cases: Either the algorithm selects an element or it does not. If the algorithm does not select any element, then $\Delta P = 0$. Moreover, the only reason the algorithm did not select an element is if $w_e\up{bx} \geq 1$ for all $e \in Q_j$. This implies $\beta_j = 0$. Since $\vec \gamma$ also does not change in this case, we have $\Delta D = 0$.

It remains to consider the case where the algorithm selects an element $e_j$ in part $j$. Let $\vec {x}' = \vec x + \one_{e_j}$ be the algorithm's allocation after $j$'s arrival. Then, on the one hand we have
$$\Delta P = b_{e_j}.$$
On the other hand, we have
\begin{align*}
    \Delta D 
    = \beta_j +  L_{f'}(\vec \gamma') - L_{f'}(\vec \gamma),
\end{align*}
where $\vec \gamma = G(\vec w\up{bx})$ and $\vec \gamma' = G(\vec{w}\upn{\vec{bx}'})$. 
Since element $e_j$ was selected, we know $$\beta_j = b_{e_j}(1 - g(w_{e_j}\up{bx})).$$ Also, 
\begin{align*}
    &L_{f'}(\vec \gamma') - L_{f'}(\vec \gamma) \\
    &= L_{f'}(G(\vec w\upn{\vec{bx}'})) - L_{f'}(G(\vec w\up{bx}))\\
    &= \iprod{\nabla_z L_{f'}(G(\vec w\up{bz})), x' - x} &\text{(for some $z \in [x, x']$, by Mean Value Theorem)}\\
    &= \iprod{bg(\vec w\up{bz}), x' - x} &\text{(by \Cref{lem:L-G-derivative})} \\
    &= b_{e_j}g(w_{e_j}\up{bz}) \\
    &\leq b_{e_j}\left(g(w_{e_j}\up{bx}) + \frac{\epsilon}{1-\epsilon} \right) &\text{(by \Cref{lem:wl-lipschitz} and $f' = (1-\epsilon)f$)}
\end{align*}
Thus,
$$\Delta D\leq b_{e_j}\left(1 + \frac{\epsilon}{1 - \epsilon}\right) = \frac{1}{1 - \epsilon}\Delta P.$$

\paragraph{Final Competitive Ratio.} To conclude, we have
$$\ALG = P \geq (1- \epsilon) D  \geq (1-\epsilon)^2\left(1 - \frac{1}{e}\right) \OPT.$$
Here, the first inequality is because $\Delta P \geq (1-\epsilon) \Delta D$ in each time step. For the second inequality, the $1 - \frac{1}{e}$ factor comes from approximate dual feasibility, and the $1-\epsilon$ factor comes from the fact that we performed the primal-dual analysis on the problem with capacities scaled down by $1-\epsilon$, which reduces the optimal value by a factor of at most $1 -\epsilon$. 
\end{proof}

\section{Online Submodular Welfare Maximization for Matroidal Utilities}
\label{sec:oswm}

In this section, we give a $(1-\nicefrac{1}{e})$-competitive algorithm for the Online Submodular Welfare Maximization problem where the utility function of each agent is the rank function of a matroid. 

Formally, there are $n$ agents, and $m$ items. The items arrive one at a time online. Each agent has an associated utility function $f_i: 2^{[m]} \to \mathbb{Z}_{\geq 0}$ which is the rank function of a matroid $\mc M_i$ on ground set $[m]$. In each time step, we must irrevocably assign the arriving item to some agent. Suppose items $U_i \subseteq [m]$ have been assigned to agents $i \in [n]$. Then the \emph{welfare} of this allocation is 

$$\sum_{i \in [n]} f_i(U_i).$$

The goal is to assign items to maximize welfare, as compared to the optimal offline allocation. We work in the value oracle model, where we can query the value $f_i(S)$ for any $i \in [n]$ and $S \subseteq [m]$ in constant time.  

Our algorithm extends to the setting where each agent has a non-negative weight $a_i$. In this case, the utility function of each agent is $f_i  := a_i \cdot \rank_{\mc{M}_i}$.


\subsection{The Matroidal Ranking Algorithm}
First, some notation. For an allocation of items to agents, we will denote by $U_i$ the set of items assigned to agent $i$. Given such an allocation, we say that an item $j$ is \emph{available} to agent $i$ if $f_i(U_i + j) > f_i(U_i)$.

The algorithm proceeds as follows. Independently for each agent $i$, select $r_i$ uniformly at random from $[0,1]$. Let the \emph{priority} of agent $i$ be defined as $a_i \cdot (1 - g(r_i))$, where $g(z) := e^{z-1}$. When an item arrives, consider the set of agents to whom this item is available, and assign the item to the highest priority agent among these. See \Cref{alg:OSWM} for a formal description.

\begin{remark}[Perusal perspective]
    We note that the Matroidal Ranking Algorithm yields the same allocation as the following procedure. In order of decreasing priority, each agent ``peruses" the full set of items in their arrival order, and greedily picks any item which increases its utility. While this perusal perspective cannot be implemented online, it yields an identical allocation as the Matroidal Ranking Algorithm, and will be useful for the analysis.
\end{remark}

\begin{algorithm}[h]
    \caption{\texttt{Matroidal Ranking Algorithm}}
    \label{alg:OSWM}
    
    \Input{An instance of Matroid-OSWM with agents $i \in [n]$, items $[m]$, and  utility functions $f_i: 2^{[m]} \to \mathbb{Z}_{\geq 0}$.}
    \Output{An allocation of items $U_i \subseteq [m]$ to each agent $i$ with welfare at least $(1-\frac{1}{e})$ times the welfare of the optimal offline allocation.}
    \BlankLine
    
    Select a value $r_i \in [0,1]$ uniformly and independently for each agent, and set the priority of agent $i$ to be $a_i(1-g(r_i))$.\;
    
    When an item $j$ arrives, assign it to the highest priority agent to whom it is available.\;
    
    \Return the resulting allocation $U_i$.\;
\end{algorithm}

\subsection{Analysis}
Consider the primal and dual problems below. The primal has variables $x_{ij}$ representing to what extent item $j \in [m]$ is allocated to agent $i \in [n]$. Notice that in the primal, rather than directly optimizing for the welfare of the agents, we simply maximize the total (weighted) quantity of items assigned, while the constraints enforce that each agent receives an independent set of items with respect to their matroid. Furthermore, there are constraints for each item enforcing that each is assigned at most once. Thus, an integer binary solution to the primal corresponds to a feasible allocation with objective value equal to the welfare of the allocation.

\begin{align*}
\max  \quad &\sum_{i \in [n]} \left(a_i \cdot \ts\sum_{j \in [m]} x_{ij}\right)  & \qquad \min \quad &\sum_{i \in [n]}\sum_{S\subseteq [m]} f_i(S)\, \alpha_{i,S} + \sum_{j\in[m]} \beta_j \\
\mbox{s.t.} \quad &\vec x(S) \leq f_i(S),  \qquad\ \, \forall i \in [n], \, S \subseteq [m] & \mbox{s.t.} \quad &\ts\sum_{S \ni j} \alpha_{i,S} + \beta_j \geq a_i, \qquad \forall i \in [n], \, j \in [m] \\
&\ts\sum_{i \in [n]}  x_{ij} \leq 1, \qquad \forall j \in [m] & &\vec \alpha, \vec \beta \geq 0  \\
&\vec x \geq 0
\end{align*}


Consider the primal solution $\vec{\overline x}$ induced by the allocation at the end of the Matroidal Ranking algorithm. This $\vec{\overline x}$ depends on the random values $r_i$ which were chosen for each agent. We will construct a dual solution $(\vec{\overline \alpha}, \vec{\overline \beta})$ whose objective value is the same as that of $\vec{\overline x}$, and which is approximately feasible in expectation. In particular, we will have

\[
    \mathbb{E}_{w \sim [0,1]^n}\left[\sum_{S \ni j} \overline\alpha_{i, S} + \overline\beta_j\right] \geq \left(\frac{e-1}{e} \right) \cdot a_i
\]
for each $(i,j) \in [n] \times [m]$. This implies that the scaled up solution $(\frac{e}{e-1})(\vec{\overline \alpha}, \vec{\overline \beta})$ is feasible in expectation, and that $\vec{\overline x}$ is $(\frac{e-1}{e})$-approximately optimal.

\paragraph{Dual Assignment} We now define the dual solution $(\vec{\overline \alpha}, \vec{\overline \beta})$. For each agent $i \in [n]$, if $U_i$ is the set of items assigned to agent $i$ at the end of the algorithm, let $S_i := \spn_{\mc M_i}(U_i)$ be the \emph{span} of $U_i$ with respect to $\mc M_i$ (i.e. the largest set of items containing $U_i$ whose rank is equal to $\rank_{\mc M_i}(U_i)$). We assign $\overline\alpha_{i,S_i} := a_i \cdot g(r_i)$. For each item $j$, if $j \in U_i$ for some $i \in [n]$, we set $\overline\beta_j := a_i \cdot (1-g(r_i))$. The remaining variables are set to zero.

\paragraph{Primal = Dual} The dual solution described above has objective value equal to the primal solution returned by the algorithm. To see this, note that whenever an item is assigned to an agent $i$, the objective value of the dual increases by exactly $a_i$. Furthermore, since we only assign an item to an agent if it is available to them, the primal objective value increases by $a_i$ as well. 

\paragraph{Expected Approximate Feasibility} We now show that the dual solution is approximately feasible in expectation. Fix a particular agent-item pair $(\spec i, \spec j)$. We will focus on the dual constraint $\sum_{S \ni \spec j} \alpha_{\spec i, S} + \beta_{\spec j} \geq a_{\spec i}$. First, condition on all random choices $r_{i}$ for $i \neq \spec i$. We denote these choices by $r_{-\spec i}$. Note that, once $r_{-\spec i}$ is fixed, the run of the algorithm is determined by the value of $r_{\spec i}$. Hence, for any choice $r_A \in [0, 1]$ we will denote a run of the algorithm in which $r_{\spec i} = r_A$ as $A$. 

For an agent $i \in [n]$ and run $A$ of the algorithm with $r_{\spec i} = r_A \in [0,1]$, let $U_{i}\upn{A, t} \subseteq [m]$ denote the set of items assigned to agent $i$ at time step $t$ of run $A$ of the algorithm. Likewise, let $U_i\upn{A}$ denote the set of items assigned to agent $i$ at the end of run $A$ of the algorithm. We will also write $\spn(U_i\upn{A, t})$ to mean $\spn_{\mc{M}_i}(U_i\upn{A, t})$, the span in agent $i$'s matroid of the set of items assigned to agent $i$ (and similarly for $\spn(U_i\upn{A})$).

We now define the \emph{critical threshold} $r^*$ to be maximum value of $r_{\spec i}$ such that $\spec j$ is in the span of the items assigned to $\spec i$ in the final allocation. Formally, 
\[
    r^* := \sup \left\{r_A \in [0,1] : \spec j \in \spn\left(U_{\spec i}\upn{A}\right)\right\}.
\]
We define $\sup (\varnothing) = 0$ by convention. 

The following key lemma characterizes several invariants that hold throughout the Matroidal Ranking algorithm. Specifically, it describes how the $\spn(U_i\upn{A, t})$ changes as $r_A$ changes. This will allow us to lower bound the expected amount of dual value assigned during the procedure.

\begin{lemma}\label{lem:OSWM_invariants}
    Fix $r_{-\spec i}$, and values $r_A$ and $r_B$ in $[0,1]$ with $r_A < r_B$. Consider the two separate runs of the algorithm: $A$ with $r_{\spec i} = r_A$ and $B$ with $r_{\spec i} = r_B$. Then at each iteration $t$, we have 
    \begin{enumerate}[(1)]
        \item $\spn(U\upn{A, t}_{\spec i}) \supseteq \spn(U\upn{B, t}_{\spec i})$,
        \item $\spn(U\upn{A, t}_{i}) = \spn(U\upn{B, t}_{i})$ for all $i \in [n]$ with $r_{i} \leq r_A$, 
        \item If $r_B = 1$, then $\spn(U\upn{A, t}_{i}) \subseteq \spn(U\upn{B, t}_{i})$ for all $i \neq \spec i$. 
    \end{enumerate}
\end{lemma}
\begin{proof}
    Points (1) and (2) follow from the perusal perspective of \Cref{alg:OSWM}. In particular, for point (1), agent $\spec i$ only peruses earlier in run $A$ than in run $B$, so $\spec i$ has more items to choose from in run $A$. For the $t^{\textrm{th}}$ item $j$, if $j \in U_{\spec i}\upn{B, t}$ and it was not already spanned by $U_{\spec i}\upn{A, t-1}$, then it would be chosen by $\spec i$ in step $t$ of run $A$. So $U_{\spec i}\upn{B, t} \subseteq \spn(U_{\spec i}\upn{A, t})$, which implies point (1). 
    
    For point (2), the perusal of all agents $i$ with $r_i < r_A$ is identical in both runs $A$ and $B$ of the algorithm, so in particular, $U\upn{A, t}_{i} = U\upn{B, t}_{i}$, implying point (2). 

    We prove point (3) by induction. Let $r_B = 1$. Suppose for induction (3) holds at iteration $t$, and a new item $j$ arrives. Consider some $i \neq \spec i$. First, if $j \in \spn(U_{i}\upn{B, t})$ already at time $t$, then we have by induction
    \[
        \spn(U_{i}\upn{A, t+1}) \subseteq \spn\left( \spn(U_{i}\upn{A, t}) \cup \{j\}\right) \subseteq \spn(U_{i}\upn{B, t+1})
    \]
    as desired. 

    So suppose otherwise that $j \not \in \spn(U\upn{B, t}_{i})$. This means that in run $B$, when item $j$ arrives, it is available to agent $i$. For the invariant in point (3) to break, item $j$ must be assigned to $i$ in run $A$ but not assigned to $i$ in run $B$. If $j$ is not assigned to $i$ in run $B$, it must be assigned to some other $i'$ with $r_{i'} < r_{i}$ (since when $j$ arrives, it is available to agent $i$). In particular, since $r_{\spec i} = 1$, we know $i' \neq i$ and we may apply induction to $i'$. This tells us that at time $t$ (before $j$'s arrival), $\spn(U\upn{A, t}_{i'}) \subseteq \spn(U\upn{B, t}_{i'})$, and therefore in run $A$, item $j$ was also available to $i'$. Since $r_{i'} < r_{i}$, this contradicts that $j$ is assigned to $i$ in run $A$. 
\end{proof}

This yields the following pair of corollaries.
\begin{corollary}\label{cor:oswm-dominance}
        If $r_{\spec i} < r^*$, then $\spec j \in \spn(U_{\spec i})$
\end{corollary}
\begin{proof}
    This follows directly from the definition of $r^*$, and point (1) from \Cref{lem:OSWM_invariants}. 
\end{proof}

\begin{corollary}\label{cor:oswm-monotonicity}
    If $r^* < 1$, then item $\spec j$ is always assigned to an agent $i$ with $r_i$ at most $r^*$.
\end{corollary}
\begin{proof}
    Observe first that for any $\eps > 0$, if $r_{\spec i} = r^* + \eps \leq 1$ then item $\spec j$ is assigned to some agent $i$ with $r_{i} < r^* + \eps$. This is because $r_{\spec i} > r^*$ implies both that item $\spec j$ is \emph{available} to $\spec i$ when it arrives, and that it is not assigned to $\spec i$. Since this holds for every $\eps > 0$ yet there are only finitely many agents, there is some such $i =: i^*$ with $r_{i^*} \leq r^*$, and some $\eps^* > 0$ such that $\spec j$ is assigned to $i^*$ when $r_{\spec i} = r^* + \eps^*$. 

    Now we claim that for any value of $r_{\spec i}$, item $\spec j$ is always available to $i^*$ when $\spec j$ arrives. This implies the claim. First, we compare instance $A$ of the algorithm with $r_{\spec i} = r_A := r^* + \eps^*$ to any instance $B$ with $r_{\spec i} = r_B > r^* + \eps^*$. By \Cref{lem:OSWM_invariants}(2) applied to $i^*$, we have $\spn(U_{i^*}\upn{A, t}) = \spn(U_{i^*}\upn{B, t})$ at the time $t$ when $\spec j$ arrives. So, in particular, $\spec j$ is available to $i^*$ in instance $B$, since it's available in instance $A$. 

    Since the above holds for any $r_B > r^* + \eps$, it in particular holds for $r_B = 1$. Now we apply \Cref{lem:OSWM_invariants}(3) to $i^*$ on any instance $A$ with $r_{\spec i} = r_A < 1$ and instance $B$ with $r_{\spec i} = r_B = 1$. We have $\spn(U_{i^*}\upn{A, t}) \subseteq \spn(U_{i^*}\upn{B, t})$ at time $t$ when $\spec j$ arrives. Therefore, again, since $i^*$ is available to $\spec j$ in instance $B$, it is also available in instance $A$. 

    So in all cases, $\spec j$ is available to $i^*$ (with $r_{i^*} \leq r^*$) when it arrives. So $\spec j$ is always assigned to an agent $i$ with $r_i$ at most $r^*$. 
\end{proof}

This gives us all ingredients required for the final proof that the Matroidal Ranking algorithm achieves a $(1 - \nicefrac{1}{e})$-competitive ratio.

\begin{proof}[Proof of \Cref{thm:matroid-oswm}]
Let $\vec{\overline x}$ be the primal solution given by the Matroidal Ranking algorithm, and $(\overline{\alpha}, \overline{\beta})$ the corresponding dual solution described above. We showed that the primal and dual objectives are equal: $\sum_{i} (a_i \cdot \sum_{j} \overline{x}_{ij}) = \sum_{i, S} f_i(S) \overline{\alpha}_{i, S} + \sum_j \overline{\beta}_j$. To argue that $\overline{x}$ is $(1 - \nicefrac{1}{e})$-competitive in expectation with the offline optimal primal solution, it then suffices by duality to show that, in expectation, $\overline{\alpha}, \overline{\beta}$ are approximately feasible. 

For any fixed agent-item pair $(\spec i, \spec j)$, we condition on the values of $r_{-\spec i}$, and let $r^*$ be the critical threshold. Then \Cref{cor:oswm-dominance} implies that 
\[
    \mathbb{E}_{r_{\spec i} \sim [0,1]}\Bigg[\sum_{S \ni \spec j} \overline\alpha_{\spec i, S}\, \bigg|\, r_{-\spec i} \Bigg] \geq \int_0^{r^*} a_{\spec i} \cdot g(z) \, dz = a_{\spec i} \cdot \left( g(r^*) - \frac{1}{e}\right).
\]
Similarly, \Cref{cor:oswm-monotonicity} implies that 
\[
    \mathbb{E}_{r_{\spec i} \sim [0,1]}\Bigg[\overline\beta_{\spec j}\, \bigg|\, r_{-\spec i} \Bigg] \geq \int_0^1 a_{\spec i}\cdot (1 - g(r^*)) \,dz = a_{\spec i}\cdot (1 - g(r^*)).
\]
(note if $r^* = 1$, then the RHS is 0, so the inequality still holds, despite \Cref{cor:oswm-monotonicity} not applying). The sum is then $a_{\spec i} \cdot (1 - \nicefrac{1}{e})$, and since this does not depend on the conditional values of $r_{-\spec i}$, we may drop the conditioning to get 
\[
    \mathbb{E}_{w \sim [0,1]^n}\Bigg[ \sum_{S \ni \spec j} \overline\alpha_{\spec i, S} + \overline\beta_{\spec j}\Bigg] \geq a_{\spec i} \cdot \left(1 - \frac{1}{e}\right)
\]
as desired. 
\end{proof}

\bibliographystyle{alpha}
\bibliography{references}

\appendix

\section{Laminar AdWords as a Case of Online SAP}
\label{sec:laminar-adwords-osap}

In the AdWords problem, we have bidders known offline. Impressions $j$ arrive online; bidder $i$ bids $b_{ij}$ dollars for this impression. Each bidder $i$ has budget $B_i$. When an impression arrives, we must irrevocably allocate it to a bidder who may afford it. Our goal is to maximize revenue. Written in terms of an linear program, 
\begin{align*}
    \max &\sum_{ij \in E} b_{ij} x_{ij} \\
    \text{s.t} \quad &\sum_{j} b_{ij} x_{ij} \leq B_i \quad \text{for all bidders $i$} \\
    &\sum_{i} x_{ij} \leq 1 \quad \text{for all impressions $j$} \\
    &x \geq 0.
\end{align*}

In the laminar setting, we account for a laminar family $\mathcal{L}$ of budget constraints on the bidders. Formally speaking, $\mathcal{L}$ is 
a laminar family over edges $E$, and each set $S \in \mathcal{L}$ has a budget constraint of $B_S$. The laminar AdWords linear program is 

\begin{align*}
\label{laminar-adwords}
    \max &\sum_{ij} b_{ij} x_{ij} \\
    \text{s.t} \quad &\sum_{ij \in S} b_{ij} x_{ij} \leq B_S \quad \text{for all $S \in \mathcal{L}$} \\
    &\sum_{i} x_{ij} \leq 1 \quad \text{for all impressions $j$} \\
    &x \geq 0.
\end{align*}

We will show a single monotone submodular function $f$ which captures these constraints. 
\begin{theorem}
    Laminar AdWords may be captured as a case of online SAP.
\end{theorem}
\begin{proof}
    Let 
\[ f(S) := \min \left\{ \sum_{T \in \mathcal{S}} B_T \colon \text{$\mathcal{S} \subseteq \mathcal{L}$ covers $S$} \right\}. \]
In other words, $f(S)$ is the most restricted budget constraint on $S$ imposed by $\mathcal{L}$. A solution satisfying the submodular assignment problem with $f$ clearly satisfies the laminar AdWords linear program. A solution to laminar AdWords linear program will also satisfy the submodular assignment problem with $f$; this follows from the fact that a minimizing sub-family $\mathcal{S}$ achieving $f(S)$ will be disjoint sets. 

We move on to showing $f$ is monotone and submodular. The former property is follows from definition. So it remains to show $f$ is submodular. Take $S, T \subseteq E$. We will show 
\[ f(S \cup T) \leq f(S) + f(T) - f(S \cap T). \]
Say $S$ and $T$ are realized by covers $\mathcal{S}$ and $\mathcal{T}$ respectively. Then, $\mathcal{S} \cup \mathcal{T}$ (here, we allow the union family to contain duplicate sets) is clearly a cover for $S \cup T$. We will show a sub-family $\mathcal{N} \subseteq \mathcal{S} \cup \mathcal{T}$ which covers $S \cap T$ and moreover, $\mathcal{S} \cup \mathcal{T} \setminus \mathcal{N}$ is still a cover for $S \cup T$. Proving this, we are left with 
\begin{align*}
    f(S \cup T) &\leq \text{budget of }(\mathcal{S} \cup \mathcal{T} \setminus \mathcal{N}) \\
    &= f(S) + f(T) - \text{budget of }(\mathcal{N}) \\
    &\leq f(S) + f(T) - f(S \cap T).
\end{align*}
So, it remains to find a set $\mathcal{N}$ satisfying (1) $\mathcal{N}$ covers $S \cap T$ and (2) $\mathcal{S} \cup \mathcal{T} \setminus \mathcal{N}$ is still a cover for $S \cup T$. Let $\mathcal{N}$ be a collection of lowest level\footnote{The laminar family is partially ordered inclusion wise, where a set $A$ is lower than $B$ if $A \subseteq B$.} sets which covers $S \cap T$. Clearly (1) is satisfied with this definition of $\mathcal{N}$. To see (2), note that elements in $S \cap T$ must be covered twice in the family $\mathcal{S} \cup \mathcal{T}$. Therefore, for any set $A \in \mathcal{N}$, there exists a set $B \in \mathcal{S} \cup \mathcal{T}$ covering $A$. This implies removing $\mathcal{N}$ from $\mathcal{S} \cup \mathcal{T}$ still leaves us with a cover for $S \cup T$. 
\end{proof}

\end{document}